\documentclass[a4paper,reqno]{amsart}
\usepackage{amsmath,amssymb,mathrsfs}
\usepackage{txfonts,bm,pifont}
\usepackage{cite}
\usepackage{float}
\usepackage{stackrel}
\usepackage{subcaption}
\usepackage{graphicx,xcolor}
\usepackage{comment}
\usepackage{enumitem}
\usepackage{longtable}
\usepackage{hyperref}

\newtheorem{theorem}{Theorem}[section]

\newtheorem{lemma}[theorem]{Lemma}
\newtheorem{remark}[theorem]{Remark}
\theoremstyle{definition}

\numberwithin{equation}{section}
\numberwithin{figure}{section}
\numberwithin{table}{section}
\newcommand{\wutilde}[1]{\vrule depth 0pt width 0pt%
{\raise0.8pt\hbox{$\smash{{\mathop{#1} \limits_{\displaystyle\widetilde{}}}}$}}}
\newcommand{\wuhat}[1]{\vrule depth 0pt width 0pt%
{\raise0.6pt\hbox{$\smash{{\mathop{#1} \limits_{\displaystyle\widehat{}}}}$}}}

\newcommand{\al}{\alpha}
\newcommand{\be}{\beta}

\newcommand{\la}{\lambda}

\newcommand{\ka}{\kappa}

\newcommand{\bbZ}{\mathbb{Z}}

\newcommand{\bbC}{\mathbb{C}}

\newcommand{\hT}{\hat{T}}
\makeatletter
\long\def\@makecaption#1#2{
 \vskip 10pt
 \setbox\@tempboxa\hbox{#1. #2}
 \ifdim \wd\@tempboxa >\hsize #1. #2\par \else \hbox
to\hsize{\hfil\box\@tempboxa\hfil}
 \fi}
\makeatother
\makeatletter
\@namedef{subjclassname@2020}{%
  \textup{2020} Mathematics Subject Classification}
\makeatother
\begin{document}
\allowdisplaybreaks

\title[]{Higher-order generalizations of the $A_6^{(1)}$- and $A_4^{(1)}$-surface type $q$-Painlev\'e equations}
\author{Nobutaka Nakazono}
\address{Institute of Engineering, Tokyo University of Agriculture and Technology, 2-24-16 Nakacho Koganei, Tokyo 184-8588, Japan.}
\email{nakazono@go.tuat.ac.jp}
\begin{abstract}
In this paper, we construct higher-order generalizations of the $A_6^{(1)}$- and $A_4^{(1)}$-surface type $q$-Painlev\'e equations from the system of partial difference equations with the consistency around a cube property by periodic reduction.
Moreover, we also show their extended affine Weyl group symmetries and Lax pairs.
\end{abstract}

\subjclass[2020]{
33E17, 
35Q53, 
37K10, 
39A13, 
39A14, 
39A23, 
39A36, 
39A45
}
\keywords{
discrete integrable systems;
$q$-Painlev\'e system;
integrable partial difference equation;
affine Weyl group;
Lax pair;
consistency around a cube property
}
\maketitle

\section{Introduction}\label{Introduction}
Fix an integer $N>0$.
This paper focuses on two $2N$-order ordinary $q$-difference equations.
One is the following:
\begin{equation}\label{eqn:intro_dP_even}
 \text{$q$P$^{(2N)}(A_6^{(1)})$ : \hspace{0.5em}}
 \overline{F}_i+\dfrac{1}{F_i}
 =\begin{cases}
  ~ \dfrac{1}{{a_i}^{2N}}\left(F_{i+1}+\dfrac{1}{\overline{F}_{i+1}}\right)
  &\text{if } i=1,\dots,2N-1,\\[1.5em]
  ~ \dfrac{\left(\displaystyle\prod_{k=1}^{2N-1}{a_k}^k\right)t}{\displaystyle\prod_{k=1}^{2N}F_k}
  &\text{if } i=2N,
 \end{cases}
\end{equation}
where $t\in\bbC$ is an independent variable, $F_i=F_i(t)\in\bbC$, $i=1,\dots,2N$, are dependent variables and $a_1,\dots,a_{2N-1}\in\bbC$ are parameters.
The symbol $\overline{\rule{0em}{0.5em}\hspace{0.4em}\,}$ denotes the discrete time evolution and
\begin{equation}
 \overline{t}=pt,\quad
 \overline{F}_i=F_i(pt),~i=1,\dots,2N,\quad
 \overline{a}_j=a_j,~j=1,\dots,2N-1,
\end{equation}
where $p\in\bbC$ is a constant parameter.
When $N=1$, the system \eqref{eqn:intro_dP_even} is equivalent to a $q$-Painlev\'e II equation of $A_6^{(1)}$-surface type\cite{RG1996:MR1399286,KTGR2000:MR1789477,RGTT2001:MR1838017,SakaiH2001:MR1882403}:
\begin{equation}\label{eqn:intro_dPII}
 \bigg(\overline{\overline{F}}_2\overline{F}_2+1\bigg)\bigg(\overline{F}_2F_2+1\bigg)
 =\dfrac{p{a_1}^3t^2\overline{F}_2}{t-a_1 \overline{F}_2}.
\end{equation}
Indeed, eliminating $F_1$ from the system \eqref{eqn:intro_dP_even} with $N=1$, we obtain Equation \eqref{eqn:intro_dPII}.

\begin{remark}
A different higher-order generalization of Equation \eqref{eqn:intro_dPII} has been reported in \cite{Okubo2017:arxiv1704.05403,MOT2021:Cluster}, but the relationship between that system and the system \eqref{eqn:intro_dP_even} is unclear.
\end{remark}

The other $2N$-order ordinary $q$-difference equation is the following:
\begin{equation}\label{eqn:intro_dP_odd}
 \text{$q$P$^{(2N)}(A_4^{(1)})$ :\hspace{0.5em}}
 \dfrac{{a_i}^{2N+1}\Big(\overline{G}_iG_i-1\Big)}{{a_i}^{2N+1}-c^{2(-1)^i}G_i}
 =\begin{cases}
 ~ \dfrac{\overline{G}_{i+1}G_{i+1}-1}{1-{a_{i+1}}^{2N+1}c^{2(-1)^i}\overline{G}_{i+1}}
 &\text{if } i=1,\dots,2N-1,\\[1.5em]
 ~ \dfrac{\left(\displaystyle\prod_{k=1}^{2N}{a_k}^k\right)ct}{\displaystyle\prod_{k=1}^N G_{2k-1}}&\text{if } i=2N,
 \end{cases}
\end{equation}
where $t$ is an independent variable, $G_i=G_i(t)\in\bbC$, $i=1,\dots,2N$, are dependent variables and $a_1,\dots,a_{2N},c\in \bbC$ are parameters.
The symbol $\overline{\rule{0em}{0.5em}\hspace{0.4em}\,}$ denotes the discrete time evolution and
\begin{equation}
 \overline{t}=pt,\quad
 \overline{G}_i=G_i(pt),~i=1,\dots,2N,\quad
 \overline{a}_j=a_j,~j=1,\dots,2N,\quad
 \overline{c}=c^{-1},
\end{equation}
where $p\in\bbC$ is a constant parameter.
When $N=1$, by using the shift operator $~\widetilde{\rule{0em}{0.5em}}~$ and the variables $g_1=g_1(t)$ and $g_2=g_2(t)$ given by
\begin{equation}
 \widetilde{\rule{0em}{0.7em}}~=\overline{\overline{\rule{0em}{0.5em}~\,~\,}},\quad
 g_1=\dfrac{c^3(t+p{a_1}^2a_2c\,G_2)}{p{a_1}^2a_2G_1},\quad
 g_2=\dfrac{a_1{a_2}^2ct+G_1}{c^4\,G_2},
\end{equation}
the system \eqref{eqn:intro_dP_odd} can be rewritten as
\begin{subequations}\label{eqn:intro_dPV}
\begin{align}
 &\Big(\widetilde{g}_1g_2-1\Big)\Big(g_1 g_2-1\Big)
 =\dfrac{t^2\Big(g_2+p{a_1}^3{a_2}^3c^{-2}\Big)\Big(p{a_1}^3{a_2}^3c^2g_1+1\Big)}{p{a_1}^4a_2\Big(a_2g_2-a_1c^{-1}t\Big)},\\
 &\Big(\widetilde{g}_2\widetilde{g}_1-1\Big)\Big(g_2 \widetilde{g}_1-1\Big)
 =\dfrac{p^2t^2\Big(\widetilde{g}_1+p{a_1}^3{a_2}^3c^2\Big)\Big(p{a_1}^3{a_2}^3c^{-2}\,\widetilde{g}_1+1\Big)}{p{a_1}^4a_2\Big(a_2\widetilde{g}_1-pa_1ct\Big)},
\end{align}
\end{subequations}
where
\begin{equation}
 \widetilde{t}=p^2t,\quad
 \widetilde{g}_1=g_1(p^2t),\quad
 \widetilde{g}_2=g_2(p^2t),\quad
 \widetilde{a}_1=a_1,\quad
 \widetilde{a}_2=a_2,\quad
 \widetilde{c}=c.
\end{equation}
The system \eqref{eqn:intro_dPV} is known as a $q$-Painlev\'e V equation of $A_4^{(1)}$-surface type \cite{TGCR2004:MR2058894,JNS2016:MR3584386,KN2015:MR3340349,NakazonoN2016:MR3503803} and can be regarded as a B\"acklund transformation of the well-known Sakai's $q$-Painlev\'e V equation of $A_4^{(1)}$-surface type\cite{SakaiH2001:MR1882403}.

\begin{remark}
For a $q$-difference equation, it is common to use the symbol ``$q$" for its shift parameter. 
However, in this paper, the symbol ``$q$" is used in the later arguments, and the relations between the symbol ``$p$" in the systems \eqref{eqn:intro_dP_even} and \eqref{eqn:intro_dP_odd} and the symbol ``$q$" in the later arguments will be given.
Therefore, we use the symbol ``$p$" instead of the symbol ``$q$" for the systems \eqref{eqn:intro_dP_even} and \eqref{eqn:intro_dP_odd} to avoid confusion.
\end{remark}

\subsection{Main results}\label{subsection:main}
In this subsection, we show four main results of this paper.
The first two results are about symmetries, while the latter two are about Lax pairs, known as showing integrability.

\begin{theorem}\label{theorem:symmetry_n=even}
The system \eqref{eqn:intro_dP_even} can be derived from a birational representation of an extended affine Weyl group of $(A_1\times A_1)^{(1)}$-type when $N=1$ and that of $(A_{2N-1}\rtimes A_1)^{(1)}$-type when $N\in\bbZ_{\geq2}$.
\end{theorem}
\begin{proof}
This theorem is proven in \S \ref{subsection:reduction_n=2} when $N=1$ and in \S \ref{subsection:reduction_n=even} when $N\in\bbZ_{\geq2}$.
\end{proof}

\begin{theorem}\label{theorem:symmetry_n=odd}
The system \eqref{eqn:intro_dP_odd} can be derived from a birational representation of an extended affine Weyl group of $(A_{2N}\rtimes A_1)^{(1)}$-type.
\end{theorem}
\begin{proof}
This theorem is proven in \S \ref{subsection:reduction_n=odd}.
\end{proof}

\begin{remark}
An affine Weyl group of $(A_1\times A_1)^{(1)}$-type means the direct product of two affine Weyl groups of $A_1^{(1)}$-type, 
while an affine Weyl group of $(A_k\rtimes A_1)^{(1)}$-type means the semi-direct product of an affine Weyl group of $A_k^{(1)}$-type and that of $A_1^{(1)}$-type.
Moreover, the extended affine Weyl group of $(A_1\times A_1)^{(1)}$-type in Theorem \ref{theorem:symmetry_n=even} indicates an affine Weyl group of $(A_1\times A_1)^{(1)}$-type extended by an automorphism of Dynkin diagrams.
On the other hand, the extended affine Weyl group of $(A_k\rtimes A_1)^{(1)}$-type in Theorems \ref{theorem:symmetry_n=even} and \ref{theorem:symmetry_n=odd} indicates an affine Weyl group of $(A_k\rtimes A_1)^{(1)}$-type extended by an automorphism of Dynkin diagrams.
(Further details can be found in Lemma \ref{lemma:lattice_Weyl_2}.)

For details on the birational representations of the extended affine Weyl groups, see
\S \ref{subsection:reduction_n=2} for the system \eqref{eqn:intro_dP_even} with $N=1$,
\S \ref{subsection:reduction_n=even} for the system \eqref{eqn:intro_dP_even} with $N\in\bbZ_{\geq2}$,
and \S \ref{subsection:reduction_n=odd} for the system \eqref{eqn:intro_dP_odd}.
\end{remark}

\begin{theorem}\label{theorem:lax_n=even}
The following pair of linear equations is a Lax pair of the system \eqref{eqn:intro_dP_even}$:$
\begin{subequations}
\begin{align}
 &\Phi(px,t)
 =\begin{pmatrix}
 \dfrac{t x}{~\displaystyle\prod_{k=1}^{2N}F_k~}&-1\\1&0
 \end{pmatrix}
 L_{2N}\dots L_1
 \Phi(x,t),\\
 &\Phi(x,pt)
 =\begin{pmatrix}
 \dfrac{t x}{~\displaystyle\prod_{k=1}^{2N}\overline{F}_k~}&-1\\1&0
 \end{pmatrix}
 \Phi(x,t),
\end{align}
\end{subequations}
where $\Phi(x,t)$ is a column vector of length two and $L_i=L_i(x,t)$, $i=1,\dots,2N$, are $2\times 2$ matrices given by
\begin{equation}
 L_i=\begin{pmatrix}
 \dfrac{\left(\displaystyle\prod_{k=1}^{2N-i}{a_{2N-k}}^k\right)xF_i}{\displaystyle\prod_{k=1}^{i-1}{a_k}^k}&-1\\
 1&-\dfrac{\left(\displaystyle\prod_{k=1}^{2N-i}{a_{2N-k}}^k\right)x}{\left(\displaystyle\prod_{k=1}^{i-1}{a_k}^k\right)F_i}
 \end{pmatrix}.
\end{equation}
\end{theorem}
\begin{proof}
This theorem is proven in \S \ref{subsection:proof_Lax_n=2} when $N=1$ and in \S \ref{subsection:proof_Lax_n=even} when $N\in\bbZ_{\geq2}$.
\end{proof}

\begin{theorem}\label{theorem:lax_n=odd}
The following pair of linear equations is a Lax pair of the system \eqref{eqn:intro_dP_odd}$:$
\begin{subequations}
\begin{align}
 &\Psi(px,t)
 =\begin{pmatrix}\dfrac{1}{c}&0\\0&\displaystyle\prod_{k=1}^NG_{2k-1}\end{pmatrix}
 M_{2N+1}K_{2N}M_{2N}\dots K_{1}M_{1}
 \begin{pmatrix}
 p^{-1}ctx&\dfrac{1}{c}\\
 c^2\left(\displaystyle\prod_{k=1}^NG_{2k}\right)&0
 \end{pmatrix}
 \Psi(x,t),\\
 &\Psi(x,pt)
 =\begin{pmatrix}
 \dfrac{1}{c}&0\\
 0&\dfrac{1}{c^4}\left(\displaystyle\prod_{k=1}^N\overline{G}_{2k}\right)+\dfrac{t}{c^3}\left(\displaystyle\prod_{k=1}^{2N}{a_k}^{k-2N-1}\right)
 \end{pmatrix}
 \begin{pmatrix}
 p^{-1}ctx&\dfrac{1}{c}\\
 c^2\left(\displaystyle\prod_{k=1}^NG_{2k}\right)&0
 \end{pmatrix}
 \Psi(x,t),
 \end{align}
\end{subequations}
where $\Psi(x,t)$ is a column vector of length two and $M_i=M_i(x,t)$, $i=1,\dots,2N+1$, and $K_j=K_j(x,t)$, $j=1,\dots,2N$, are $2\times 2$ matrices given by
\begin{equation}
 M_i=\begin{pmatrix}
 \dfrac{\displaystyle\prod_{k=i}^{2N}{a_k}^{2N+1}}{\displaystyle\prod_{k=1}^{2N}{a_k}^k}x&1\\
 1&\dfrac{\displaystyle\prod_{k=i}^{2N}{a_k}^{2N+1}}{\displaystyle\prod_{k=1}^{2N}{a_k}^k}x
 \end{pmatrix},\quad
 K_j=\begin{pmatrix}c^{2(-1)^j}&0\\0&-\dfrac{1}{G_j}\end{pmatrix}.
\end{equation}
\end{theorem}
\begin{proof}
This theorem is proven in \S \ref{subsection:proof_Lax_n=odd}.
\end{proof}

\subsection{Background}\label{subsection:Background}
In the early 20th century, to find a new class of special functions, Painlev\'e and Gambier classified all the ordinary differential equations of the type $y''=F(y',y,t)$, where $y=y(t)$, $'=d/dt$ and $F$ is a function meromorphic in $t$ and rational in $y$ and $y'$, 
with the Painlev\'e property (solutions do not have movable singularities other than poles) \cite{PainleveP1902:MR1554937,GambierB1910:MR1555055}.
As a result, they obtained six new equations, and the resulting equations are now collectively referred to as the Painlev\'e equations.
Note that the Painlev\'e VI equation was found by Fuchs \cite{FuchsR1905:quelques} before Painlev\'e and Gambier.

After the discovery, the Painlev\'e equations withdrew from the stage of ``modern mathematics" for a while.
The Painlev\'e equations regained attention after the 1970s because they appeared in mathematical physics research.
For instance, solutions of the Painlev\'e equations (Painlev\'e transcendents) were rediscovered as scaling functions for the two-dimensional Ising model on a square lattice \cite{WMTB1976:PhysRevB.13.316} and as similarity solutions of the soliton equations describing specific shallow water waves (solitons) \cite{AS1977:PhysRevLett.38.1103}.

$q$-Painlev\'e equations are a family of second-order nonlinear ordinary $q$-difference equations.
Historically, they have been obtained as $q$-discrete analogues of the Painlev\'e equations (see, for example, \cite{GRP1991:MR1125950,JS1996:MR1403067}).
Similar to the Painlev\'e equations, $q$-Painlev\'e equations are also known to describe special solutions of various discrete soliton equations \cite{GRSWC2005:MR2117991,HHJN2007:MR2303490,FJN2008:MR2425981,OrmerodCM2012:MR2997166,
HHNS2015:MR3317164}.
Particularly famous is imposing periodic conditions on soliton equations, which are closely related to the result in this paper.
It is also known that $q$-Painlev\'e equations have the same good properties as the Painlev\'e equations, such as having Lax pairs and (extended) affine Weyl group symmetries.
(See for example \cite{KNY2017:MR3609039,HJN2016:MR3587455}.)

Let us explain more about a Lax pair and an (extended) affine Weyl group symmetry relevant to this paper.
A Lax pair of a $q$-Painlev\'e equation is given by a pair of two linear equations
\begin{equation}
 \phi(qx,t)=A(x,t)\phi(x,t),\quad
 \phi(x,qt)=B(x,t)\phi(x,t),
\end{equation}
such that their compatibility condition 
\begin{equation}
 B(qx,t)A(x,t)=A(x,qt)B(x,t)
\end{equation}
gives the corresponding $q$-Painlev\'e equation. 
Here, $\phi(x,t)$ is a column vector, $A(x,t)$ and $B(x,t)$ are square matrices, $x\in\bbC$ is a parameter, and $t\in\bbC$ and $q\in\bbC$ are respectively the independent variable and the shift parameter of the $q$-Painlev\'e equation.
Next, we explain an (extended) affine Weyl group symmetry of a $q$-Painlev\'e equation.
A transformation from an integrable system to an integrable system is called a B\"acklund transformation. 
In the case of a $q$-Painlev\'e equation, there are self-B\"acklund transformations, that is, B\"acklund transformations to itself, which collectively form an (extended) affine Weyl group.
In that case, the $q$-Painlev\'e equation is said to have an (extended) affine Weyl group symmetry.
Note that an (extended) affine Weyl group symmetry of a $q$-Painlev\'e equation does not always mean its B\"acklund transformations alone. It often means a large group of transformations, including its time evolution.
For example, Theorem \ref{theorem:symmetry_n=even} asserts that the system \eqref{eqn:intro_dP_even} can be obtained from a birational action of an extended affine Weyl group of $(A_{2N-1}\rtimes A_1)^{(1)}$-type.
In this case, we say that the system \eqref{eqn:intro_dP_even} has an extended affine Weyl group symmetry of $(A_{2N-1}\rtimes A_1)^{(1)}$-type.

There exist six Painlev\'e equations.
However, by Okamoto's space of initial values \cite{OkamotoK1979:MR614694}, the Painlev\'e III equation can be classified into three types, and then the Painlev\'e equations can be considered as eight types \cite{OKSO2006:MR2277519}.
On the other hand, an infinite number of $q$-Painlev\'e equations exist.
By considering Sakai's space of initial values \cite{SakaiH2001:MR1882403}, which is an extension of Okamoto's space of initial values, $q$-Painlev\'e equations can be classified into nine surface types (see Figure \ref{fig:Sakai_surface}).

\begin{figure}[htbp]
\[
 A_0^{(1)\ast}
 \to A_1^{(1)}
 \to A_2^{(1)}
 \to A_3^{(1)}
 \to A_4^{(1)}
 \to A_5^{(1)}
 \to A_6^{(1)}~
 \begin{matrix}
 \nearrow\\
 \searrow
 \end{matrix}
 \begin{matrix}
 ~A_7^{(1)}\\[1.6em]
 ~{A_7^{(1)}}'
 \end{matrix}
\]
\caption{
Types of spaces of initial values for $q$-Painlev\'e equations.
The surface degenerates in the direction of the arrow due to the specialization and confluence of the base points that characterize the surface types.
On each surface, transformations collectively forming an (extended) affine Weyl group exist. A birational action of an element with infinite order of each transformation group gives rise to a $q$-Painlev\'e equation.
}
\label{fig:Sakai_surface}
\end{figure}

In \cite{KNY2002:MR1958118}, Kajiwara-Noumi-Yamada showed a birational representation of the extended affine Weyl group of $(A_{m-1}\times A_{n-1})^{(1)}$-type (KNY's representation), where $m$ and $n$ are integers greater than or equal to $2$, except for $(m,n)=(2,2)$.
Note that KNY's representation is essentially the same even if $m$ and $n$ are interchanged \cite{NY2018:zbMATH06876428}.
It was shown in \cite{KNY2002:MR1917133} that KNY's representation gives the Painlev\'e type $q$-difference equations, including $q$-Painlev\'e equations as the second-order ordinary $q$-difference equations.
Indeed, \cite{KNY2002:MR1917133} showed that the case $(m,n)=(2,3)$ gives $q$-Painlev\'e equations of $A_5^{(1)}$-surface type, and \cite{KNY2002:MR1917133,TakenawaT2003:MR1996297} showed the case $(m,n)=(2,4)$ gives $q$-Painlev\'e equations of $A_3^{(1)}$-surface type.

Recently, it has been reported that KNY's representation can be extended to the birational representation of the extended affine Weyl group of $(A_{m-1}\times A_{n-1}\times A_{g-1})^{(1)}$-type (extended KNY's representation), where $g$ is the common greatest divisor of $m$ and $n$ \cite{MOT2018:AmAnAg,MOT2023:AmAnAgArxiv}.
The paper \cite{OS2020:10.1093/imrn/rnaa283} presents the explicit forms of the Painlev\'e type $q$-difference equations obtained from the extended KNY's representation.
It is also shown that these equations include Sakai's $q$-Garnier system \cite{SakaiH2005:MR2177121}, Tsuda's $q$-Painlev\'e system arising from the $q$-UC hierarchy \cite{TsudaT2010:MR2563787} and Suzuki's $q$-Painlev\'e system arising from the $q$-DS hierarchy \cite{Suzuki2015:zbMATH06608012,Suzuki2017:zbMATH06944877}.
Moreover, in \cite{SuzukiT2019:1523388080554396672}, Suzuki claims that the case $(m,n,g)=(3,3,3)$ gives $q$-Painlev\'e equations of $A_2^{(1)}$-surface type.

As mentioned above, KNY's representation derived in 2002 is still being studied from various angles and is undergoing further development.
The motivation for this study is to find research subjects that would be studied for a long time and used in a wide range of fields, as KNY's representation is.

We consider that the birational representation in this study and the extended KNY's representation are different for the following reasons:
\begin{itemize}
\item 
Different periodic conditions are imposed on the same system for their derivations.

In our previous works \cite{JNS2015:MR3403054,JNS2014:MR3291391}, under a periodic condition, a system of partial difference equations having the consistency around a cube (CAC) property was found to give a birational representation of an extended affine Weyl group of type $(A_2\times A_1)^{(1)}$, which is KNY's representation with $(m,n)=(2,3)$.
By comparing the results in \cite{JNS2015:MR3403054,JNS2014:MR3291391} and those in \cite{KNY2002:MR1958118,KNY2002:MR1917133}, it is expected that by imposing the periodic condition
\begin{equation}\label{eqn:intro_period_cond_2}
 U(l_1+1,\dots,l_n+1,l_0)=U(l_1,\dots,l_n,l_0)
\end{equation}
on the system of partial difference equations given in \S \ref{section:CAC_system}, KNY's representation with $m=2$ can be obtained.

On the other hand, by imposing the different periodic condition
\begin{equation}\label{eqn:intro_period_cond_1}
 U(l_1+1,\dots,l_n+1,l_0+1)=U(l_1,\dots,l_n,l_0)
\end{equation}
on the same system, we can obtain a birational representation studied in this paper.
The case $n=2$ is studied in \cite{JNS2015:MR3403054}, and its birational representation gives $q$-Painlev\'e equations of $A_6^{(1)}$-surface type.
Moreover, the case $n=3$ is studied in \cite{JNS2016:MR3584386}, and its birational representation gives $q$-Painlev\'e equations of $A_4^{(1)}$-surface type.
The case $n\geq 4$ is a new result of this paper.
\item 
From the viewpoint of cluster algebra\cite{MOT2021:Cluster}, the extended KNY's representation does not give $q$-Painlev\'e equations of $A_7^{(1)}$- and $A_6^{(1)}$-surface type.
However, the resulting birational representation in this study gives Equation \eqref{eqn:intro_dPII}, which is a $q$-Painlev\'e equation of $A_6^{(1)}$-surface type.
\end{itemize}

\begin{remark}
In \cite{AJT2020:zbMATH07212273}, higher-order generalizations of the $A_5^{(1)}$- and $A_3^{(1)}$-surface type $q$-Painlev\'e equations were derived by imposing the $(n,1)$-type periodic condition{\rm :}
\begin{equation}
 w(l_1+n,l_2+1)=w(l_1,l_2)
\end{equation}
on the multi-parametric version of the discrete modified KdV equation{\rm :}
\begin{equation}\label{eqn:intro_mpver_lmkdv}
 \dfrac{w(l_1+1,l_2+1)}{w(l_1,l_2)}
 =\dfrac{\al^{(1)}(l_1)\, w(l_1,l_2+1)-\be^{(1)}(l_2) w(l_1+1,l_2)}{\al^{(2)}(l_1)\, w(l_1+1,l_2)-\be^{(2)}(l_2) w(l_1,l_2+1)},
\end{equation}
where $l_1,l_2\in\bbZ$ are lattice parameters and $\{\al^{(1)}(l),\al^{(2)}(l),\be^{(1)}(l),\be^{(2)}(l)\}_{l\in\bbZ}$ are complex parameters.
From the following facts, the higher-order $q$-Painlev\'e equations in \cite{AJT2020:zbMATH07212273} are expected to be obtained from the (extended) KNY's representation.
\begin{itemize}
\item 
When 
\begin{equation}
 \al^{(1)}(l)=\al^{(2)}(l),\quad
 \be^{(1)}(l)=\be^{(2)}(l),
\end{equation}
Equation \eqref{eqn:intro_mpver_lmkdv} is called the lattice modified KdV (lmKdV) equation \cite{NC1995:MR1329559,NQC1983:MR719638}.
\item
In \cite{JNS2015:MR3403054}, using the case $n=2$ as an example, it was shown that the $(n,1)$-type periodically reduced lmKdV equation could be obtained from the $\omega$-lattice constructed from the $(1,1,\dots,1)$-type periodic reduction of the $(n+1)$-dimensional system of lmKdV equations by considering its restricted lattice. 
\item 
The papers \cite{JNS2015:MR3403054,JNS2014:MR3291391}, together with the papers \cite{KNY2002:MR1958118,KNY2002:MR1917133}, indicate that the $(1,1,\dots,1)$-type periodic reduction of the system of lmKdV equations gives KNY's representation with $m=2$.
\end{itemize}
\end{remark}

\subsection{Notation and Terminology}\label{subsection:notation_definitions}
This paper will use the following notations and terminologies for conciseness.
\begin{itemize}
\item 
For matrices $A$ and $B$, the symbol $AB$ means the matrix product $A.B$.
\item 
For transformations $s$ and $r$, the symbol $sr$ means the composite transformation $s\circ r$.
\item 
The ``$1$" in the transformation implies the identity transformation.
\item
For a transformation $s$, the relation $s^\infty=1$ means that there is no positive integer $k$ such as
$s^k=1$.
\item 
If the subscript number is greater than the superscript number in the product symbol, $1$ is assumed.
For example,
\begin{equation}
 \prod_{k=1}^0 2^k=1.
\end{equation}
\end{itemize}
\subsection{Outline of the paper}
This paper is organized as follows.
In \S \ref{section:CAC_system}, we introduce a system of partial difference equations with the CAC property and show its properties, a Lax representation, and transformations that keep the system invariant.
In \S \ref{section:reductions} and \S \ref{section:Lax_eqns}, 
we give proofs of Theorems \ref{theorem:symmetry_n=even} and \ref{theorem:symmetry_n=odd} and those of Theorems \ref{theorem:lax_n=even} and \ref{theorem:lax_n=odd}, respectively.
Some concluding remarks are given in \S \ref{ConcludingRemarks}.
\section{A system of partial difference equations with the CAC property}\label{section:CAC_system}
Fix an integer $n\geq2$.
We consider the following system of partial difference equations:
\begin{subequations}\label{eqns:ABS_U}
\begin{align}
 &\la(l_{0\cdots n})^2\dfrac{U_{\overline{ij}}}{U}=\dfrac{\al^{(i)}(l_i) U_{\overline{i}}-\al^{(j)}(l_j) U_{\overline{j}}}{\al^{(j)}(l_j) U_{\overline{i}}-\al^{(i)}(l_i) U_{\overline{j}}},
 &&i<j,\quad i,j\in\{1,\dots,n\},\\
 &\dfrac{U_{\overline{0\,k}}}{U}+\la(l_{0\cdots n})^4\dfrac{U_{\overline{0}}}{U_{\overline{k}}}+\al^{(k)}(l_k) \ka(l_0) \la(l_{0\cdots n})=0,
 &&k=1,\dots,n,\label{eqn:ABS_U_0k}
\end{align}
\end{subequations}
where $l_0,\dots,l_n\in\bbZ$ are lattice parameters, 
$\{\al^{(1)}(l),\dots,\al^{(n)}(l),\ka(l),\la(0)\}_{l\in\bbZ}$ are complex parameters and
\begin{subequations}
\begin{align}
 &U=U(l_1,\dots,l_n,l_0),\quad
 U_{\overline{i}}=U|_{\,l_i\to l_i+1},\quad
 U_{\overline{ij}}=U|_{(l_i,l_j)\to (l_i+1,l_j+1)},\\
 &l_{0\cdots n}:=\sum_{i=0}^nl_i,\quad
 \la(l)=\begin{cases}
 \la(0)&l\in2\bbZ,\\[0.5em]
 \dfrac{1}{\la(0)}&\text{otherwise}.
 \end{cases} 
\end{align}
\end{subequations}
\begin{remark}
By letting
\begin{equation}
 u(l_1,\dots,l_n,l_0)=H(l_1+\dots+l_n,l_0)U(l_1,\dots,l_n,l_0),
\end{equation}
where $H(l,l_0)$ satisfies
\begin{equation}
 H(l+1,l_0)=\dfrac{1}{\la(l+l_0)H(l,l_0)},\quad
 H(l,l_0+1)=\dfrac{\la(l+l_0)^2}{H(l,l_0)},
\end{equation}
the system \eqref{eqns:ABS_U} can be rewritten as
\begin{subequations}\label{eqns:ABS_u}
\begin{align}
 &\dfrac{u_{\overline{ij}}}{u}=\dfrac{\al^{(i)}(l_i) u_{\overline{i}}-\al^{(j)}(l_j) u_{\overline{j}}}{\al^{(j)}(l_j) u_{\overline{i}}-\al^{(i)}(l_i) u_{\overline{j}}},
 &&i<j,\quad i,j\in\{1,\dots,n\},\label{eqn:ABS_u_H3}\\
 &\dfrac{u_{\overline{0k}}}{u}+\dfrac{u_{\overline{0}}}{u_{\overline{k}}}+\al^{(k)}(l_k) \ka(l_0)=0,
 &&k=1,\dots,n,\label{eqn:ABS_u_D4}
\end{align}
\end{subequations}
where
\begin{equation}
 u=u(l_1,\dots,l_n,l_0),\quad
 u_{\overline{i}}=u|_{\,l_i\to l_i+1},\quad
 u_{\overline{ij}}=u|_{(l_i,l_j)\to (l_i+1,l_j+1)}.
\end{equation} 
We can easily verify that the system \eqref{eqns:ABS_u} has the consistency around a cube (CAC) property, which is known as a type of integrability (see \cite{BS2002:MR1890049,NijhoffFW2002:MR1912127,WalkerAJ:thesis,NW2001:MR1869690} for the CAC property).
Equation \eqref{eqn:ABS_u_H3} with fixed $i$ and $j$ is known as the lattice modified KdV equation \cite{NC1995:MR1329559,NQC1983:MR719638} or the H3 equation in the study of the CAC property\cite{ABS2003:MR1962121,ABS2009:MR2503862}, while Equation \eqref{eqn:ABS_u_D4} with fixed $k$ is known as Boll's D4 equation\cite{BollR2011:MR2846098,BollR2012:MR3010833}.
\end{remark}

The remainder of this section will discuss the properties of the system \eqref{eqns:ABS_U}.

\subsection{Symmetry of the system \eqref{eqns:ABS_U}}
We define the automorphisms of the lattice $\bbZ^{n+1}$:
$s_1$, \dots, $s_{n-1}$, $\pi$, $w_1$,
by the following actions on the coordinates $(l_1,\dots,l_n,l_0)\in\bbZ^{n+1}$:
\begin{subequations}
\begin{align}
 s_i~:&(l_1,\dots,l_n,l_0)\mapsto (l_1,\dots,l_n,l_0)\Big|_{l_i\leftrightarrow l_{i+1}},\quad
 i=1,\dots,n-1,\\
 \pi~:&(l_1,\dots,l_n,l_0)\mapsto (l_n+1,l_1,\dots,l_{n-1},l_0),\\
 w_1:&(l_1,\dots,l_n,l_0)\mapsto (-l_n,\dots,-l_1,-l_0-1).
\end{align}
\end{subequations}
We lift the action of these transformations to the action on the parameters and the $U$ variable in the system \eqref{eqns:ABS_U} by
\begin{subequations}\label{eqns:trans_general_U}
\begin{align}
 &s_i(\al^{(j)}(l))
 =\begin{cases}
 \al^{(i+1)}(l)&\text{if } j=i,\\
 \al^{(i)}(l)&\text{if } j=i+1,\\
 \al^{(j)}(l)&\text{otherwise},
 \end{cases}\qquad
 s_i(U)=U|_{l_i\leftrightarrow l_{i+1}},\quad
 i=1,\dots,n-1,\\
 &\pi(\al^{(j)}(l))
 =\begin{cases}
 \al^{(j+1)}(l)&\text{if } j=1,\dots,n-1,\\
 \al^{(1)}(l+1)&\text{if } j=n,
 \end{cases}\qquad
 \pi(\la(l))=\dfrac{1}{\la(l)},\\
 &\pi(U)=U(l_n+1,l_1,\dots,l_{n-1},l_0),\\
 &w_1(\al^{(j)}(l))=\al^{(n+1-j)}(-l-1),~ j=1,\dots,n,\quad
 w_1(\ka(l))=\ka(-l-2),\\
 &w_1(\la(l))=\dfrac{1}{\la(l)},\quad
 w_1(U)=\dfrac{1}{U(-l_n,\dots,-l_1,-l_0-1)},
\end{align}
\end{subequations}
where $U=U(l_1,\dots,l_n,l_0)$.
Note that the lifting of the action of such transformations can be easily deduced from the fact that the system \eqref{eqns:ABS_U} is a system of partial difference equations placed on the lattice $\bbZ^{n+1}$ with the CAC property.
(See for example \cite{JNS2016:MR3584386,NakazonoN2018:MR3760161}.)

Define the transformations $s_0$, $w_0$ and $r$ by
\begin{equation}\label{eqn:s0w0w1rdef}
 s_0=\pi^{-1}s_1\pi,\quad
 w_0=\pi^2w_1,\quad
 r=\pi w_1.
\end{equation}
Then, the following lemma holds.
\begin{lemma}\label{lemma:lattice_Weyl}
~\\[-1.5em]
\begin{itemize}
\item[{\bf (i)\,:}] 
The system \eqref{eqns:ABS_U} is invariant under the action of $\langle s_1,\dots,s_{n-1},w_1,\pi\rangle$.
\item[{\bf (ii)\,:}]
$\langle s_1,\dots,s_{n-1},w_1,\pi\rangle=\langle s_0,\dots,s_{n-1},w_0,w_1,r\rangle$. 
\item[{\bf (iii)\,:}] 
The groups of transformations $\langle s_0,\dots,s_{n-1}\rangle$ and $\langle w_0,w_1\rangle$ form the affine Weyl groups of type $A_{n-1}^{(1)}$ and $A_1^{(1)}$, respectively.
Moreover, when $n=2$,
\begin{subequations}
\begin{align}
 &\langle s_0,s_1,w_0,w_1\rangle=\langle s_0,s_1\rangle\times \langle w_0,w_1\rangle,\\
 &\langle s_0,s_1,r\rangle=\langle s_0,s_1\rangle\rtimes\langle r\rangle,\quad
 \langle w_0,w_1,r\rangle=\langle w_0,w_1\rangle\rtimes\langle r\rangle
\end{align}
\end{subequations}
hold, while $n>2$,
\begin{subequations}
\begin{align}
 &\langle s_0,\dots,s_{n-1},w_0,w_1\rangle=\langle s_0,\dots,s_{n-1}\rangle\rtimes\langle w_0,w_1\rangle,\\
 &\langle s_0,\dots,s_{n-1},r\rangle=\langle s_0,\dots,s_{n-1}\rangle\rtimes\langle r\rangle,\quad
 \langle w_0,w_1,r\rangle=\langle w_0,w_1\rangle\rtimes\langle r\rangle
\end{align}
\end{subequations}
hold.
\end{itemize}
\end{lemma}
\begin{proof}
{\bf (i)} can be verified by direct calculation.
Let us consider {\bf (ii)} and {\bf (iii)} separately for $n=2$ and $n>2$.
\begin{description}
\item[Case $n=2$]
We can verify that the transformations $\{s_1,w_1,\pi\}$ satisfy the following relations:
\begin{equation}\label{eqn:rels_s_pi_io_n=2}
 {s_1}^2=w_1^2=\pi^\infty=1,\quad
 (w_1 s_1)^2=(\pi s_1)^\infty=(\pi w_1)^2=1,\quad
 (\pi^2 s_1)^2=1.
\end{equation}
The definition \eqref{eqn:s0w0w1rdef} gives 
\begin{equation}
 \langle s_1,w_1,\pi\rangle\supset\langle s_0,s_1,w_0,w_1,r\rangle,
\end{equation}
while the relation
\begin{equation}
 \pi=\pi {w_1}^2=(\pi w_1)w_1=rw_1
\end{equation}
gives
\begin{equation}
 \langle s_1,w_1,\pi\rangle\subset\langle s_0,s_1,w_0,w_1,r\rangle.
\end{equation}
Therefore, {\bf (ii)} holds.
The relation \eqref{eqn:rels_s_pi_io_n=2} leads to the fact that the transformations $\{s_0,s_1,w_0,w_1\}$ satisfy the following relations:
\begin{subequations}\label{eqns:fundamental_A1A1}
\begin{align}
 &{s_0}^2={s_1}^2=(s_0s_1)^\infty=1,\quad
 {w_0}^2={w_1}^2=(w_0w_1)^\infty=1,\\
 &(w_0s_0)^2=(w_0s_1)^2=(w_1s_0)^2=(w_1s_1)^2=1,
\end{align}
and the transformation $r$ satisfies the following relations:
\begin{equation}\label{eqn:fundamental_A1A1_2}
 r^2=1,\quad
 rs_0=s_1r,\quad
 rs_1=s_0r,\quad
 rw_0=w_1r,\quad
 rw_1=w_0r.
\end{equation}
\end{subequations}
Therefore, {\bf (iii)} holds.
\item[Case $n>2$]
The discussion is the same as in the case $n=2$, so the details are omitted.
The transformations $\{s_0,\dots,s_{n-1},w_1,\pi\}$ satisfy the following relations:
\begin{subequations}
\begin{align}
 &{s_i}^2=(s_is_{i\pm 1})^3=(s_is_j)^2=1,~ j\neq i\pm 1,\quad
 {w_1}^2=1,\\
 &w_1 s_k=s_{n-k}w_1,\quad
 \pi^\infty=1,\quad
 \pi s_k=s_{k+1}\pi,\quad
 (\pi w_1)^2=1,
\end{align}
\end{subequations}
where $i,j,k\in\bbZ/n\bbZ$.
Note that to simplify the description of relations, $s_0$ is added to the transformations to be considered above.
Moreover, the transformations $\{s_0,\dots,s_{n-1},w_0,w_1\}$ satisfy the following relations:
\begin{subequations}\label{eqns:fundamental_An-1A1}
\begin{align}
 &{s_i}^2=(s_is_{i\pm 1})^3=(s_is_j)^2=1,~ j\neq i\pm 1,\quad
 {w_0}^2={w_1}^2=(w_0w_1)^\infty=1,\\
 &w_0s_k=s_{n-k+2}w_0,\quad
 w_1s_k=s_{n-k}w_1,
\end{align}
where $i,j,k\in\bbZ/n\bbZ$,
and the transformation $r$ satisfies the following relations:
\begin{equation}\label{eqn:fundamental_An-1A1_2}
 r^2=1,\quad
 rs_i=s_{n-i+1}r,\quad
 rw_0=w_1r,\quad
 rw_1=w_0r,
\end{equation}
\end{subequations}
where $i\in\bbZ/n\bbZ$.
\end{description}
\end{proof}

Lemma \ref{lemma:lattice_Weyl} leads the following lemma.

\begin{lemma}\label{lemma:lattice_Weyl_2}
The following hold:
\begin{itemize}
\item[{\bf (i)\,:}] 
When $n=2$, the transformation group $\langle s_0,s_1,w_0,w_1,r\rangle$ satisfies the relations \eqref{eqns:fundamental_A1A1}.
Both transformation groups $\langle s_0,s_1\rangle$ and $\langle w_0,w_1\rangle$ form the
affine Weyl group of type $A_1^{(1)}$.
Therefore, we denote the transformation group $\langle s_0,s_1,w_0,w_1\rangle=\langle s_0,s_1\rangle\times\langle w_0,w_1\rangle$ as $W((A_1\times A_1)^{(1)})$.
Moreover, as evident from the relations \eqref{eqn:fundamental_A1A1_2}, the transformation $r$ corresponds to reflections of two Dynkin diagrams of type $A_1^{(1)}$ associated with $\langle s_0,s_1\rangle$ and $\langle w_0,w_1\rangle$.
Therefore, we refer to the transformation group $\langle s_0,s_1,w_0,w_1,r\rangle=W((A_1\times A_1)^{(1)})\rtimes\langle r\rangle$ as an extended affine Weyl group of type $(A_1\times A_1)^{(1)}$
and denote it as $\widetilde{W}((A_1\times A_1)^{(1)})$.
\item[{\bf (ii)\,:}] 
When $n>2$, the transformation group $\langle s_0,\dots,s_{n-1},w_0,w_1,r\rangle$ satisfies the relations \eqref{eqns:fundamental_An-1A1}.
Transformation groups $\langle s_0,\dots,s_{n-1}\rangle$ and $\langle w_0,w_1\rangle$ form the affine Weyl group of type $A_{n-1}^{(1)}$ and that of type $A_1^{(1)}$, respectively.
Therefore, we denote the transformation group $\langle s_0,\dots,s_{n-1},w_0,w_1\rangle=\langle s_0,\dots,s_{n-1}\rangle\rtimes\langle w_0,w_1\rangle$ as $W((A_{n-1}\rtimes A_1)^{(1)})$.
Moreover, as evident from the relations \eqref{eqn:fundamental_An-1A1_2}, the transformation $r$ corresponds to reflections of the Dynkin diagram of type $A_{n-1}^{(1)}$ associated with $\langle s_0,\dots,s_{n-1}\rangle$ and that of type $A_1^{(1)}$ associated with $\langle w_0,w_1\rangle$.
Therefore, we refer to the transformation group $\langle s_0,\dots,s_{n-1},w_0,w_1,r\rangle=W((A_{n-1}\rtimes A_1)^{(1)})\rtimes\langle r\rangle$ as an extended affine Weyl group of type $(A_{n-1}\rtimes A_1)^{(1)}$
and denote it as $\widetilde{W}((A_{n-1}\rtimes A_1)^{(1)})$.

Note that we sometimes write $\widetilde{W}((A_{n-1}\rtimes A_1)^{(1)})$, including the case $n=2$, for simplicity.
\item[{\bf (iii)\,:}]
The system \eqref{eqns:ABS_U} is invariant under the action of $\widetilde{W}((A_{n-1}\rtimes A_1)^{(1)})$.
\end{itemize}
\end{lemma}

\begin{remark}\label{remark:Ti_U}
We lift the action of automorphisms of the lattice $\bbZ^{n+1}$$:$
\begin{equation}
 T_i:(l_1,\dots,l_n,l_0)\mapsto (l_1,\dots,l_n,l_0)\Big|_{l_i\to l_{i}+1},\quad
 i=0,\dots,n,
\end{equation}
to the action on the parameters and the $U$ variable in the system \eqref{eqns:ABS_U} by
\begin{subequations}
\begin{align}
 &T_i(\al^{(j)}(l))
 =\begin{cases}
 \al^{(i)}(l+1)&\text{if } j=i,\\
 \al^{(j)}(l)&\text{otherwise},\quad
 \end{cases}\quad
 T_i(\kappa(l))
 =\begin{cases}
 \kappa(l+1)&\text{if } i=0,\\
 \kappa(l)&\text{otherwise},\quad
 \end{cases}\\
 &T_i(\la(l))=\la(l+1),\quad
 T_i(U)=U_{\overline{i}},
\end{align}
\end{subequations}
where $U=U(l_1,\dots,l_n,l_0)$.
Each $T_i$ can be regarded as the shift operator in the $l_i$-direction of the system \eqref{eqns:ABS_U}.
Moreover, $T_1$, \dots, $T_n$ can be expressed by the composite transformations of $\widetilde{W}((A_{n-1}\rtimes A_1)^{(1)})$ as
\begin{equation}\label{eqn:Ti_Weyl}
 T_1=\pi s_{n-1}s_{n-2}\cdots s_1,\quad
 T_{i+1}=\pi T_i\pi^{-1},~i=1,\dots,n-1,
\end{equation}
but not $T_0$.
This is because $T_0$ acts on the parameter $\ka(l)$ as $T_0(\ka(l))=\ka(l+1)$, whereas the elements of $\widetilde{W}((A_{n-1}\rtimes A_1)^{(1)})$ have no such action.
\end{remark}

\subsection{Lax representation of the system \eqref{eqns:ABS_U}}
We obtain a Lax representation of the system \eqref{eqns:ABS_U} following the method given in \cite{BS2002:MR1890049,NijhoffFW2002:MR1912127,WalkerAJ:thesis} as follows.
\begin{subequations}\label{eqns:Weyl_phi}
\begin{align}
 &\phi_{\overline{i}}
 =\begin{pmatrix}
 \dfrac{\mu U_{\overline{i}}}{\al^{(i)}(l_i)U}&-\dfrac{{U_{\overline{i}}}^2}{\la(l_{0\cdots n})}\\[1em]
 \dfrac{\la(l_{0\cdots n})}{U^2}&-\dfrac{\mu U_{\overline{i}}}{\al^{(i)}(l_i)U}
 \end{pmatrix}\phi,\quad
 i=1,\dots,n,\\
 &\phi_{\overline{0}}
 =\begin{pmatrix}
 -\dfrac{\mu\ka(l_0)U_{\overline{0}}}{U}&-\la(l_{0\cdots n})^2{U_{\overline{0}}}^2\\[1em]
 \dfrac{1}{\la(l_{0\cdots n})^2U^2}&0
 \end{pmatrix}\phi,
\end{align}
\end{subequations}
where $\mu\in\bbC$ is a spectral variable, $\phi=\phi(l_1,\dots,l_n,l_0)$ is a column vector of length two and 
\begin{equation}
 \phi_{\overline{i}}=\phi|_{\,l_i\to l_i+1}.
\end{equation}
Indeed, we can easily verify that the compatibility conditions
\begin{equation}
 \Big(\phi_{\overline{i}}\Big)_{\overline{j}}=\Big(\phi_{\overline{j}}\Big)_{\overline{i}}~,\quad
 0\leq i<j\leq n,
\end{equation}
give the system \eqref{eqns:ABS_U}.

\begin{lemma}
Let the action of $\widetilde{W}((A_{n-1}\rtimes A_1)^{(1)})$ on the vector $\phi=\phi(l_1,\dots,l_n,l_0)$ be given by
\begin{subequations}
\begin{align}
 &s_i(\phi)=\phi|_{l_i\leftrightarrow l_{i+1}},\quad
  i=1,\dots,n-1,\\
 &\pi(\phi)=\phi(l_n+1,l_1,\dots,l_{n-1},l_0),\\
 &w_1(\phi)=\dfrac{\displaystyle\prod_{k=1}^n A_k(-l_{n+1-k})}{U(-l_n,\dots,-l_1,-l_0-1)^2}\begin{pmatrix}0&1\\1&0\end{pmatrix}\phi(-l_n,\dots,-l_1,-l_0-1),
\end{align}
\end{subequations}
where $A_i(l)$, $i=1,\dots,n$, satisfies
\begin{equation}
 \dfrac{A_i(l+1)}{A_i(l)}=\dfrac{\al^{(i)}(l)^2}{\al^{(i)}(l)^2-\mu^2}.
\end{equation}
Moreover, the transformations $s_0$, $w_0$ and $r$ are defined by \eqref{eqn:s0w0w1rdef} and the parameter $\mu$ is invariant under the action of $\widetilde{W}((A_{n-1}\rtimes A_1)^{(1)})$.
Then, the system \eqref{eqns:Weyl_phi} is invariant under the action of $\widetilde{W}((A_{n-1}\rtimes A_1)^{(1)})$.
\end{lemma}
\begin{proof}
It can be verified by direct calculation.
\end{proof}

\begin{remark}\label{remark:Ti_phi}
Let the action of $T_i$, $i=0,\dots,n$, on the vector $\phi=\phi(l_1,\dots,l_n,l_0)$ and the parameter $\mu$ be given by
\begin{equation}
 T_i(\phi)=\phi_{\overline{i}},\quad
 T_i(\mu)=\mu.
\end{equation}
Similar to Remark \ref{remark:Ti_U}, $T_i$, $i=1,\dots,n$, can be expressed as \eqref{eqn:Ti_Weyl} by using elements of $\widetilde{W}((A_{n-1}\rtimes A_1)^{(1)})$, even if the action on the vector $\phi=\phi(l_1,\dots,l_n,l_0)$ and the parameter $\mu$ are included.
\end{remark}

\section{Reduction from the system \eqref{eqns:ABS_U} to the $q$P$^{(2N)}(A_6^{(1)})$ \eqref{eqn:intro_dP_even} and the $q$P$^{(2N)}(A_4^{(1)})$ \eqref{eqn:intro_dP_odd}}\label{section:reductions}
In this section, we derive the systems \eqref{eqn:intro_dP_even} and \eqref{eqn:intro_dP_odd} by imposing the periodic condition 
\begin{equation}
 U(l_1+1,\dots,l_n+1,l_0+1)=U(l_1,\dots,l_n,l_0)
\end{equation}
on the system \eqref{eqns:ABS_U}.
Moreover, we also give proofs of Theorems \ref{theorem:symmetry_n=even} and \ref{theorem:symmetry_n=odd}.

Since the conditions for the parameters obtained by the periodic reduction depend on the parity of $n$ in the system \eqref{eqns:ABS_U}, we consider the reduction in the cases of odd and even $n$ separately.
In addition, since $n=2$ differs from other even-numbered cases, we further separate them.
Note that the results for $n=2$ and $n=3$ are the same as those in \cite{JNS2015:MR3403054} and \cite{JNS2016:MR3584386}, respectively.

\subsection{Proof of Theorem \ref{theorem:symmetry_n=even} with $N=1$}\label{subsection:reduction_n=2}
Let us consider the periodic reduction of the system \eqref{eqns:ABS_U} when $n=2$.
Imposing the $(1,1,1)$-periodic condition
\begin{equation}\label{eqn:period_U_2}
 U(l_1+1,l_2+1,l_0+1)=U(l_1,l_2,l_0),
\end{equation}
on the system \eqref{eqns:ABS_U}, we obtain the following conditions of the parameters:
\begin{equation}
 \dfrac{\al^{(1)}(l_1+1)}{\al^{(1)}(l_1)}
 =\dfrac{\al^{(2)}(l_2+1)}{\al^{(2)}(l_2)}
 =\dfrac{\ka(l_0)}{\ka(l_0+1)},\quad
 \la(l)^4=1.
\end{equation}
Therefore, let
\begin{equation}\label{eqn:period_para_2}
 \al^{(1)}(l)=q^{2l}\al^{(1)}(0),\quad
 \al^{(2)}(l)=q^{2l}\al^{(2)}(0),\quad
 \ka(l)=q^{-2l}\ka(0),\quad
 \la(l)=1,
\end{equation}
where $q\in\bbC$ is a parameter.
Define three parameters $\{a_0,a_1,b\}$ and two variables $\{f_1,f_2\}$ by
\begin{subequations}
\begin{align}
 &a_0=\dfrac{\al^{(1)}(1)^{1/2}}{\al^{(2)}(0)^{1/2}},\quad
 a_1=\dfrac{\al^{(2)}(0)^{1/2}}{\al^{(1)}(0)^{1/2}},\quad
 b=\al^{(1)}(0)^{1/2}\al^{(2)}(0)^{1/2}\ka(0),\\
 &f_1=\dfrac{U(0,0,0)}{U(1,0,0)},\quad 
 f_2=\dfrac{U(1,0,0)}{U(1,1,0)}.
\end{align}
\end{subequations}
Then, the following holds:
\begin{equation}
 a_0a_1=q.
\end{equation}
From the action \eqref{eqns:trans_general_U}, the action of $\widetilde{W}((A_1\times A_1)^{(1)})$ on the new parameters and $f$ variables is obtained as the following lemma.

\begin{lemma}\label{lemma:n=2_f_Weyl}
The action of $\widetilde{W}((A_1\times A_1)^{(1)})=\langle s_0,s_1,w_0,w_1\rangle\rtimes\langle r\rangle$ on the parameters $\{a_0,a_1,b,q\}$ and the variables $\{f_1,f_2\}$ is given by
\begin{subequations}\label{eqns:WA1_para_f}
\begin{align}
 s_1:\,&(a_0,a_1,b,q,f_1,f_2)
 \to\left(a_0{a_1}^2,{a_1}^{-1},b,q,f_1\dfrac{{a_1}^2f_1f_2-1}{f_1f_2-{a_1}^2},\,f_2\dfrac{f_1f_2-{a_1}^2}{{a_1}^2f_1f_2-1}\right),\\
 \pi:\,&(a_0,a_1,b,q,f_1,f_2)
 \to\left(a_1,a_0,qb,q,f_2,\,-\dfrac{1}{f_1}\left(1+\dfrac{qa_0b}{f_2}\right)\,\right),\\
 w_1:\,&(a_0,a_1,b,q,f_1,f_2)
 \to\left({a_0}^{-1},{a_1}^{-1},q^2b,q^{-1},f_2,f_1\right).
\end{align}
\end{subequations}
Note that the transformations $s_0$, $w_0$ and $r$ are defined by \eqref{eqn:s0w0w1rdef}.
Under the action above, the elements of $\widetilde{W}((A_1\times A_1)^{(1)})$ also satisfy the relations \eqref{eqns:fundamental_A1A1}.
\end{lemma}
\begin{proof}
The action on the parameters is obvious.
Let us consider the action on the $f$ variables.
Define the three variables $\{\omega_0,\omega_1,\omega_2\}$ by
\begin{equation}\label{eqn:def_omega_2}
 \omega_0=U(0,0,0),\quad
 \omega_1=U(1,0,0),\quad
 \omega_2=U(1,1,0).
\end{equation}
Using the system \eqref{eqns:ABS_U}, we obtain
\begin{align}
 &s_1(\omega_1)
 =U(0,1,0)
 =U(1,0,0)\cfrac{\dfrac{\al^{(2)}(0)}{\al^{(1)}(0)}U(1,1,0)-U(0,0,0)}{U(1,1,0)-\dfrac{\al^{(2)}(0)}{\al^{(1)}(0)}U(0,0,0)}
 =\omega_1\dfrac{{a_1}^2\omega_2-\omega_0}{\omega_2-{a_1}^2\omega_0},\\
 &\pi(\omega_2)
 =U(2,1,0)
 =-\dfrac{U(1,1,1)U(1,1,0)}{\al^{(1)}(1)\ka(0)U(1,1,0)+U(2,1,1)}
 =-\dfrac{\omega_0\omega_2}{qa_0b\omega_2+\omega_1},
\end{align}
and thereby, we have
\begin{subequations}
\begin{align}
 &s_1:(\omega_0,\omega_1,\omega_2)\to
 \left(\omega_0,\,\omega_1\dfrac{{a_1}^2\omega_2-\omega_0}{\omega_2-{a_1}^2\omega_0},\,\omega_2\right),\\
 &\pi:(\omega_0,\omega_1,\omega_2)\to
 \left(\omega_1,\,\omega_2,\,-\dfrac{\omega_0\omega_2}{qa_0b\omega_2+\omega_1}\right),\\
 &w_1:(\omega_0,\omega_1,\omega_2)\to
 \left(\dfrac{1}{\omega_2},\,\dfrac{1}{\omega_1},\,\dfrac{1}{\omega_0}\right).
\end{align}
\end{subequations}
Then, the statement follows from 
\begin{equation}
 f_1=\dfrac{\omega_0}{\omega_1},\quad 
 f_2=\dfrac{\omega_1}{\omega_2}.
\end{equation}
\end{proof}

\begin{remark}
The action of $\widetilde{W}((A_1\times A_1)^{(1)})=\langle s_0,s_1,w_0,w_1\rangle\rtimes\langle r\rangle$ on the parameters $\{a_0,a_1,b,q\}$ and the variables $\{f_1,f_2\}$ given by \eqref{eqns:WA1_para_f} corresponds to the action of $\widetilde{\bm W}((A_1\times A_1')^{(1)})=\langle {\bm s}_0,{\bm s}_1,{\bm w}_0,{\bm w}_1\rangle\rtimes\langle{\bm \pi}\rangle$ on the parameters $\{{\bm a}_0,{\bm a}_1,{\bm b},{\bm q}\}$ and the variables $\{{\bm f}_0,{\bm f}_1,{\bm f}_2\}$ given in \cite{JNS2015:MR3403054} by the following correspondence:
\begin{subequations}
\begin{align}
 &{\bm s}_0=s_0,\quad
 {\bm s}_1=s_1,\quad
 {\bm w}_0=w_0,\quad
 {\bm w}_1=w_1,\quad
 {\bm \pi}=r,\\
 &{\bm a}_0={a_0}^4,\quad
 {\bm a}_1={a_1}^4,\quad
 {\bm b}=-{a_0}^6{a_1}^4b^2,\quad
 {\bm q}=q^4,\\
 &{\bm f}_0={a_0}^{-2}{a_1}^{-1}b^{-1}f_2,\quad
 {\bm f}_1={a_1}^{-1}b^{-1}f_1,\quad
 {\bm f}_2=\dfrac{{a_0}^2{a_1}^2b^2}{f_1f_2}.
\end{align}
\end{subequations}
As shown in \cite{JNS2015:MR3403054}, $\widetilde{\bm W}((A_1\times A_1')^{(1)})$ is the extended affine Weyl symmetry group for Sakai's $A_6^{(1)}$-surface, which gives the $q$-Painlev\'e II equation \eqref{eqn:intro_dPII}.
Note that $\widetilde{\bm W}((A_1\times A_1')^{(1)})$ is an extension of the direct product of the transformation groups $\langle {\bm s}_0,{\bm s}_1\rangle$ and $\langle {\bm w}_0,{\bm w}_1\rangle$ by the transformation ${\bm \pi}$. 
Furthermore, both $\langle {\bm s}_0,{\bm s}_1\rangle$ and $\langle {\bm w}_0,{\bm w}_1\rangle$ form the affine Weyl group of type $A_1^{(1)}$, and the transformation ${\bm \pi}$ corresponds to reflections of the two Dynkin diagrams of type $A_1^{(1)}$ associated with $\langle {\bm s}_0,{\bm s}_1\rangle$ and $\langle {\bm w}_0,{\bm w}_1\rangle$.
\end{remark}

Let us give the transformations corresponding to the shift operators $T_1$, $T_2$, $T_0$ in Remark \ref{remark:Ti_U} with $n=2$ using the elements of $\widetilde{W}((A_1\times A_1)^{(1)})$.
As noted in \eqref{eqn:Ti_Weyl}, $T_1$ and $T_2$ are given by
\begin{equation}
 T_1=\pi s_1,\quad
 T_2=\pi s_0,
\end{equation}
whose actions on the parameters $\{a_0,a_1,b,q\}$ are given by
\begin{subequations}
\begin{align}
 &T_1:(a_0,a_1,b,q)\to(qa_0,q^{-1}a_1,qb,q),\\
 &T_2:(a_0,a_1,b,q)\to(q^{-1}a_0,qa_1,qb,q).
\end{align}
\end{subequations}
As explained in Remark \ref{remark:Ti_U}, there is no element in $\widetilde{W}((A_1\times A_1)^{(1)})$ that fully corresponds to $T_0$.
Let
\begin{equation}
 \hT_0=\pi^{-2}.
\end{equation}
Its action on the periodically reduced $U$ variable is given by
\begin{equation}
 \hT_0(U(l_1,l_2,l_0))=U(l_1-1,l_2-1,l_0)=U(l_1,l_2,l_0+1),
\end{equation}
which is the same as the action of $T_0$.
Moreover, the actions of $T_0$ and $\hT_0$ on the parameters $\{a_0,a_1,b,q\}$ are also same as shown below.
\begin{subequations}\label{eqns:T0_para_n=2}
\begin{align}
 &\hT_0:(a_0,a_1,b,q)\to(a_0,a_1,q^{-2}b,q),\\
 &T_0:(a_0,a_1,b,q)\to(a_0,a_1,q^{-2}b,q).
\end{align}
\end{subequations}
Therefore, in what follows, we will not distinguish between $T_0$ and $\hT_0$ when considering their actions on the parameters $\{a_0,a_1,b,q\}$ and the $f$ variables.

\begin{remark}
Note that the action on the parameters $\big\{\al^{(1)}(0),\al^{(2)}(0),\ka(0)\big\}$ is different for $T_0$ and $\hT_0$ as follows.
\begin{subequations}
\begin{align}
 &\hT_0:\big(\al^{(1)}(0),\al^{(2)}(0),\ka(0)\big)\to\big(\al^{(1)}(-1),\al^{(2)}(-1),\ka(0)\big),\\
 &T_0:\big(\al^{(1)}(0),\al^{(2)}(0),\ka(0)\big)\to\big(\al^{(1)}(0),\al^{(2)}(0),\ka(1)\big).
\end{align}
\end{subequations}
\end{remark}

The action of $\hT_0$ (or $T_0$) on the $f$ variables is given by
\begin{equation}\label{eqn:T0_f_n=2}
 \hT_0(f_1)+\dfrac{1}{f_1}=\dfrac{1}{{a_1}^2}\left(f_2+\dfrac{1}{\hT_0(f_2)}\right),\quad
 \hT_0(f_2)+\dfrac{1}{f_2}=-\dfrac{a_1b}{f_1f_2},
\end{equation}
which is equivalent to the system \eqref{eqn:intro_dP_even} with $N=1$ by the following correspondence:
\begin{equation}\label{eqn:correspondence_introEqn_n=2}
 \overline{\rule{0em}{0.5em}~\,}=\hT_0,\quad
 F_i= f_i,\quad
 t=-b,\quad
 p=q^{-2}.
\end{equation}
Theorem \ref{theorem:symmetry_n=even} with $N=1$ follows from this fact and Lemma \ref{lemma:n=2_f_Weyl}.

\begin{remark}
We here consider the $q$-Painlev\'e equation given by $\hT_0$.
However, from the actions of $T_1$ and $T_2$, we can also obtain different $q$-Painlev\'e equations (see, for example \cite{JNS2015:MR3403054,nakazono2022discrete}).
In general, various discrete dynamical systems of Painlev\'e type can be obtained from elements of infinite order, not necessarily translations, in an (extended) affine Weyl group\cite{KNT2011:MR2773334,KN2015:MR3340349}.
\end{remark}

\subsection{Proof of Theorem \ref{theorem:symmetry_n=even} with $N\in\bbZ_{\geq2}$}\label{subsection:reduction_n=even}
In this subsection, we consider the periodic reduction of the system \eqref{eqns:ABS_U} in the case $n=2N$ $(N=2,3,\dots)$.
The discussion is the same as in \S \ref{subsection:reduction_n=2}, so the details are omitted.

Imposing the periodic condition
\begin{equation}\label{eqn:period_U_even}
 U(l_1+1,\dots,l_{2N}+1,l_0+1)=U(l_1,\dots,l_{2N},l_0),
\end{equation}
on the system \eqref{eqns:ABS_U}, we obtain the following conditions:
\begin{equation}
 \dfrac{\al^{(1)}(l_1+1)}{\al^{(1)}(l_1)}
 =\cdots
 =\dfrac{\al^{(2N)}(l_{2N}+1)}{\al^{(2N)}(l_{2N})}
 =\dfrac{\ka(l_0)}{\ka(l_0+1)},\quad
 \la(l)^4=1.
\end{equation}
Therefore, let
\begin{equation}\label{eqn:period_para_even}
 \al^{(i)}(l)=q^{2Nl}\al^{(i)}(0),\quad
 i=1,\dots,2N,\quad
 \ka(l)=q^{-2Nl}\ka(0),\quad
 \la(l)=1,
\end{equation}
where $q\in\bbC$ is a parameter.
Define the parameters $\{a_0,\dots,a_{2N-1},b\}$ and the variables $\{f_1,\dots,f_{2N}\}$ by
\begin{subequations}\label{eqns:def_para_f_even}
\begin{align}
 &a_0=\dfrac{\al^{(1)}(1)^{1/(2N)}}{\al^{(2N)}(0)^{1/(2N)}},\quad
 a_i=\dfrac{\al^{(i+1)}(0)^{1/(2N)}}{\al^{(i)}(0)^{1/(2N)}},~
 i=1,\dots,2N-1,\\
 &b=\left(\prod_{k=1}^{2N}\al^{(k)}(0)^{1/(2N)}\right)\ka(0),\quad
 f_j=\dfrac{\omega_{j-1}}{\omega_j},~
 j=1,\dots,2N,
\end{align}
\end{subequations}
where
\begin{equation}\label{eqn:def_omega_even}
 \omega_0=U(0,\dots,0),\quad
 \omega_i=\omega_{i-1}|_{l_i\to l_i+1},~ 
 i=1,\dots,2N.
\end{equation}
Then, the following holds:
\begin{equation}
 \prod_{i=0}^{2N-1}a_i=q.
\end{equation}
From the action \eqref{eqns:trans_general_U}, the action of $\widetilde{W}((A_{2N-1}\rtimes A_1)^{(1)})$ on the new parameters and $f$ variables is obtained as the following lemma.

\begin{lemma}\label{lemma:n=even_f_Weyl}
The action of $\widetilde{W}((A_{2N-1}\rtimes A_1)^{(1)})=\langle s_0,\dots,s_{2N-1},w_0,w_1\rangle\rtimes\langle r\rangle$ on the parameters $\{a_0,\dots,a_{2N-1},b,q\}$ is given by
\begin{subequations}\label{eqns:WA2N-1_para}
\begin{align}
 &s_i(a_j)
 =\begin{cases}
 {a_i}^{-1}&\text{if } j=i,\\
 a_ia_{i\pm 1}&\text{if } j=i\pm 1,\\
 a_j&\text{otherwise},
 \end{cases}\qquad
 s_i(b)=b,\quad
 s_i(q)=q,\\
 &\pi(a_j)=a_{j+1},\quad
 \pi(b)=qb,\quad
 \pi(q)=q,\\
 &w_1(a_j)=\dfrac{1}{a_{2N-j}},\quad
 w_1(b)=q^{2N}b,\quad
 w_1(q)=q^{-1},
\end{align}
\end{subequations}
where $i,j\in\bbZ/(2N)\bbZ$,
while that on the variables $\{f_1,\dots,f_{2N}\}$ is given by
\begin{subequations}\label{eqns:WA2N-1_f}
\begin{align}
 &s_i(f_j)
 =\begin{cases}
 f_i\dfrac{1-{a_i}^{2N}f_if_{i+1}}{{a_i}^{2N}-f_if_{i+1}}&\text{if } j=i,\\
 f_{i+1}\dfrac{{a_i}^{2N}-f_if_{i+1}}{1-{a_i}^{2N}f_if_{i+1}}&\text{if } j=i+1,\\
 f_j&\text{otherwise},
 \end{cases}\qquad
 i=1,\dots,2N-1,~
 j=1,\dots,2N,\\
 &\pi(f_j)
 =\begin{cases}
 f_{j+1}&\text{if } j=1,\dots,2N-1,\\
 -\dfrac{q^{2N-1}a_0b}{\left(\displaystyle\prod_{k=1}^{2(N-1)}{a_k}^{2N-k-1}\right)\left(\displaystyle\prod_{k=1}^{2N}f_k\right)}-\dfrac{1}{f_1}&\text{if } j=2N,
 \end{cases}\\
 &w_1(f_j)=f_{2N-j+1},\quad
 j=1,\dots,2N.
\end{align}
\end{subequations}
Note that the transformations $s_0$, $w_0$ and $r$ are defined by \eqref{eqn:s0w0w1rdef}.
Under the action above, the relations \eqref{eqns:fundamental_An-1A1} with $n=2N$ hold.
\end{lemma}
\begin{proof}
The action on the parameters is obvious.
Therefore, we only consider the action on the $f$ variables.
The action on the variables $\{\omega_0,\dots,\omega_{2N}\}$ given in \eqref{eqn:def_omega_even} is given by
\begin{subequations}
\begin{align}
 &s_i(\omega_j)
 =\begin{cases}
 \omega_i\dfrac{{a_i}^{2N}\omega_{i+1}-\omega_{i-1}}{\omega_{i+1}-{a_i}^{2N}\omega_{i-1}}&\text{if } j=i,\\
 \omega_j&\text{otherwise},
 \end{cases}\quad
 i=1,\dots,2N-1,~
 j=0,\dots,2N,\\
 &\pi(\omega_j)
 =\begin{cases}
 \omega_{j+1}&\text{if } j=0,\dots,2N-1,\\
 -\dfrac{\left(\displaystyle\prod_{k=1}^{2(N-1)}{a_k}^{2N-k-1}\right)\omega_0\omega_{2N}}{q^{2N-1}a_0b\omega_{2N}+\left(\displaystyle\prod_{k=1}^{2(N-1)}{a_k}^{2N-k-1}\right)\omega_1}&\text{if } j=2N,
 \end{cases}\\
 &w_1(\omega_j)=\dfrac{1}{\omega_{2N-j}},\quad 
 j=0,\dots,2N.
\end{align}
\end{subequations}
Then, the statement follows from the relation between the $f$ variables and the $\omega$ variables given in \eqref{eqns:def_para_f_even}.
\end{proof}

Let us define the transformations $T_i$, $i=1,\dots,2N$, by \eqref{eqn:Ti_Weyl} and the transformation $\hT_0$ by
\begin{equation}
 \hT_0=\pi^{-2N}.
\end{equation}
The action of $T_1,\dots,T_{2N},\hT_0$ on the parameters $\{a_0,\dots,a_{2N-1},b,q\}$ is given by
\begin{subequations}
\begin{align}
 &T_i(a_j)
 =\begin{cases}
 qa_{i-1}&\text{if } j=i-1\quad ({\rm mod}~2N),\\
 q^{-1}a_i&\text{if } j=i\quad ({\rm mod}~2N),\\
 a_j&\text{otherwise},
 \end{cases}\qquad
 i=1,\dots,2N,\\
 &T_i:(b,q)\to(q b,q),\quad
 i=1,\dots,2N,\\
 &\hT_0:(a_0,\dots,a_{2N-1},b,q)\to(a_0,\dots,a_{2N-1},q^{-2N}b,q).
\end{align}
\end{subequations}

\begin{lemma}
The following holds:
\begin{equation}\label{eqn:dP_even}
 \hT_0(f_i)+\dfrac{1}{f_i}
 =\begin{cases}
  ~ \dfrac{1}{{a_i}^{2N}}\left(f_{i+1}+\dfrac{1}{\hT_0(f_{i+1})}\right)
  &\text{if } i=1,\dots,2N-1,\\[1em]
  ~ -\dfrac{\left(\displaystyle\prod_{k=1}^{2N-1}{a_k}^k\right)b}{\displaystyle\prod_{k=1}^{2N}f_k}
  &\text{if } i=2N.
 \end{cases}
\end{equation}
\end{lemma}
\begin{proof}
The following holds:
\begin{equation}
 {\hT_0}^{~-1}(f_1)+\dfrac{1}{f_1}
 =\pi(f_{2N})+\dfrac{1}{f_1}
 =-\dfrac{q^{2N-1}a_0b}{\left(\displaystyle\prod_{k=1}^{2(N-1)}{a_k}^{2N-k-1}\right)\left(\displaystyle\prod_{k=1}^{2N}f_k\right)}.
\end{equation}
Applying $\hT_0$ to the equation above, we obtain
\begin{equation}\label{eqn:dP_2N_proof_1}
 f_1+\dfrac{1}{\hT_0(f_1)}
 =-\dfrac{q^{-1}a_0b}{\left(\displaystyle\prod_{k=1}^{2(N-1)}{a_k}^{2N-k-1}\right)\left(\displaystyle\prod_{k=1}^{2N}\hT_0(f_k)\right)}.
\end{equation}
Moreover, applying the transformation $\pi^{2N-1}$ to \eqref{eqn:dP_2N_proof_1}, we obtain
\begin{equation}
 f_{2N}+\dfrac{1}{\hT_0(f_{2N})}
 =-\dfrac{q^{2N-2}a_{2N-1}b}{\left(\displaystyle\prod_{k=1}^{2(N-1)}{a_{k-1}}^{2N-k-1}\right)\left(\displaystyle\prod_{k=1}^{2N-1}f_k\right)\hT_0(f_{2N})},
\end{equation}
which gives
\begin{equation}
 \hT_0(f_{2N})+\dfrac{1}{f_{2N}}
 =-\dfrac{q^{2N-1}b}{\left(\displaystyle\prod_{k=0}^{2(N-1)}{a_k}^{2N-k-1}\right)\left(\displaystyle\prod_{k=1}^{2N}f_k\right)}
 =-\dfrac{\left(\displaystyle\prod_{k=1}^{2N-1}{a_k}^k\right)b}{\displaystyle\prod_{k=1}^{2N}f_k}.
\end{equation}
The equation above is Equation \eqref{eqn:dP_even} when $i=2N$.
Furthermore, applying the transformation $\pi$ to \eqref{eqn:dP_2N_proof_1}, we obtain
\begin{align}
 f_2+\dfrac{1}{\hT_0(f_2)}
 &=-\dfrac{a_1b}{\left(\displaystyle\prod_{k=1}^{2(N-1)}{a_{k+1}}^{2N-k-1}\right)\left(\displaystyle\prod_{k=1}^{2N-1}\hT_0(f_{k+1})\right)f_1}\notag\\
 &=-\dfrac{a_1b\hT_0(f_1)}{\left(\displaystyle\prod_{k=1}^{2(N-1)}{a_{k+1}}^{2N-k-1}\right)\left(\displaystyle\prod_{k=1}^{2N}\hT_0(f_k)\right)f_1}\notag\\
 &={a_1}^{2N}\left(\hT_0(f_1)+\dfrac{1}{f_1}\right),
\end{align}
which gives
\begin{equation}\label{eqn:eqn:dP_2N_proof_2}
 \hT_0(f_1)+\dfrac{1}{f_1}=\dfrac{1}{{a_1}^{2N}}\left(f_2+\dfrac{1}{\hT_0(f_2)}\right).
\end{equation}
The equation above is Equation \eqref{eqn:dP_even} when $i=1$.
Then the statement follows by applying the transformations $\pi^m$, $m=1,\dots,N-2$, to Equation \eqref{eqn:eqn:dP_2N_proof_2}.
\end{proof}

The system \eqref{eqn:dP_even} is equivalent to the system \eqref{eqn:intro_dP_even} with $N\in\bbZ_{\geq2}$ by the following correspondence:
\begin{equation}\label{eqn:correspondence_introEqn_n=even}
 \overline{\rule{0em}{0.5em}~\,}=\hT_0,\quad
 F_i=f_i,\quad
 t=-b,\quad
 p=q^{-2N}.
\end{equation}
Therefore, from this fact and Lemma \ref{lemma:n=even_f_Weyl}, Theorem \ref{theorem:symmetry_n=even} with $N\in\bbZ_{\geq2}$ holds.

\subsection{Proof of Theorem \ref{theorem:symmetry_n=odd}}\label{subsection:reduction_n=odd}
In this subsection, we consider the periodic reduction of the system \eqref{eqns:ABS_U} in the case $n=2N+1$ $(N=1,2,\dots)$.
The discussion is the same as in \S \ref{subsection:reduction_n=2}, so the details are omitted.

Imposing the periodic condition
\begin{equation}\label{eqn:period_U_odd}
 U(l_1+1,\dots,l_{2N+1}+1,l_0+1)=U(l_1,\dots,l_{2N+1},l_0),
\end{equation}
on the system \eqref{eqns:ABS_U}, we obtain the following conditions:
\begin{equation}
 \dfrac{\al^{(1)}(l_1+1)}{\al^{(1)}(l_1)}
 =\cdots
 =\dfrac{\al^{(2N+1)}(l_{2N+1}+1)}{\al^{(2N+1)}(l_{2N+1})}
 =\dfrac{\ka(l_0)}{\ka(l_0+1)}.
\end{equation}
Therefore, let
\begin{equation}\label{eqn:period_para_odd}
 \al^{(i)}(l)=q^{(2N+1)l}\al^{(i)}(0),\quad
 i=1,\dots,2N+1,\quad
 \ka(l)=q^{-(2N+1)l}\ka(0),
\end{equation}
where $q\in\bbC$ is a parameter.
Define the parameters $\{a_0,\dots,a_{2N},b,c\}$ and the variables $\{f_1,\dots,f_{2N}\}$ by
\begin{subequations}\label{eqns:def_para_f_odd}
\begin{align}
 &a_0=\dfrac{\al^{(1)}(1)^{1/(2N+1)}}{\al^{(2N+1)}(0)^{1/(2N+1)}},\quad
 a_i=\dfrac{\al^{(i+1)}(0)^{1/(2N+1)}}{\al^{(i)}(0)^{1/(2N+1)}},~
 i=1,\dots,2N,\\
 &b=\left(\prod_{k=1}^{2N+1}\al^{(k)}(0)^{1/(2N+1)}\right)\ka(0),\quad
 \la(l)
 =\begin{cases}
 c&\text{if } l\in2\bbZ,\\
 c^{-1}&\text{otherwise},
 \end{cases}\\
 &f_j=\dfrac{\omega_{j-1}}{\omega_{j+1}},~
 j=1,\dots,2N,
\end{align}
\end{subequations}
where
\begin{equation}\label{eqn:def_omega_odd}
 \omega_0=U(0,\dots,0),\quad
 \omega_i=\omega_{i-1}|_{l_i\to l_i+1},~
 i=1,\dots,2N+1.
\end{equation}
Then, the following holds:
\begin{equation}
 \prod_{i=0}^{2N}a_i=q.
\end{equation}
From the action \eqref{eqns:trans_general_U}, the action of $\widetilde{W}((A_{2N}\rtimes A_1)^{(1)})$ on the new parameters and $f$ variables is obtained as the following lemma.

\begin{lemma}\label{lemma:n=odd_f_Weyl}
The action of $\widetilde{W}((A_{2N}\rtimes A_1)^{(1)})=\langle s_0,\dots,s_{2N},w_0,w_1\rangle\rtimes\langle r\rangle$ on the parameters $\{a_0,\dots,a_{2N},b,c,q\}$ is given by
\begin{subequations}\label{eqns:WA2N_para}
\begin{align}
 &s_i(a_j)
 =\begin{cases}
 {a_i}^{-1}&\text{if } j=i,\\
 a_ia_{i\pm 1}&\text{if } j=i\pm 1,\\
 a_j&\text{otherwise},
 \end{cases}\qquad
 s_i(b)=b,\quad
 s_i(c)=c,\quad
 s_i(q)=q,\\
 &\pi(a_i)=a_{i+1},\quad
 \pi(b)=qb,\quad
 \pi(c)=c^{-1},\quad
 \pi(q)=q,\\
 &w_1(a_i)=\dfrac{1}{a_{2N+1-i}},\quad
 w_1(b)=q^{2N+1}b,\quad
 w_1(c)=c^{-1},\quad
 w_1(q)=q^{-1},
\end{align}
\end{subequations}
where $i,j\in\bbZ/(2N+1)\bbZ$,
while that on the variables $\{f_1,\dots,f_{2N}\}$ is given by
\begin{subequations}\label{eqns:WA2N_f}
\begin{align}
 &s_i(f_j)
 =\begin{cases}
 f_{i-1}\dfrac{\la(i-1)^2-{a_i}^{2N+1}f_i}{{a_i}^{2N+1}\la(i-1)^2-f_i}&\text{if } j=i-1,\\
 f_{i+1}\dfrac{{a_i}^{2N+1}\la(i-1)^2-f_i}{\la(i-1)^2-{a_i}^{2N+1}f_i}&\text{if } j=i+1,\\
 f_j&\text{otherwise},
 \end{cases}\qquad
 i,j=1,\dots,2N,\\
 &\pi(f_j)
 =\begin{cases}
 f_{j+1}&\text{if } j=1,\dots,2N-1,\\
 -\dfrac{c^4}{\displaystyle\prod_{k=1}^N f_{2k-1}}
 \left(
 \displaystyle\prod_{k=1}^N f_{2k}+\dfrac{q^{2N}a_0b}{\left(\displaystyle\prod_{k=1}^{2N-1}{a_k}^{2N-k}\right)c}
 \right)&\text{if } j=2N,
 \end{cases}\\
 &w_1(f_j)=f_{2N+1-j},\quad
 j=1,\dots,2N.
\end{align}
\end{subequations}
Note that the transformations $s_0$, $w_0$ and $r$ are defined by \eqref{eqn:s0w0w1rdef}.
Under the action above, the relations \eqref{eqns:fundamental_An-1A1} with $n=2N+1$ hold.
\end{lemma}
\begin{proof}
The action on the parameters is obvious.
Therefore, we only consider the action on the $f$ variables.
The action on the variables $\{\omega_0,\dots,\omega_{2N+1}\}$ given in \eqref{eqn:def_omega_odd} is given by
\begin{subequations}
\begin{align}
&s_i(\omega_j)
 =\begin{cases}
 \omega_i\dfrac{{a_i}^{2N+1}\la(i-1)^2\omega_{i+1}-\omega_{i-1}}{\la(i-1)^2\omega_{i+1}-{a_i}^{2N+1}\omega_{i-1}}&\text{if } j=i,\\
 \omega_j&\text{otherwise},
 \end{cases}\notag\\
 &\quad
 i=1,\dots,2N,~
 j=0,\dots,2N+1,\\
 &\pi(\omega_j)
 =\begin{cases}
 \omega_{j+1}&\text{if } j=0,\dots,2N,\\
 -\dfrac{\left(\displaystyle\prod_{k=1}^{2N-1}{a_k}^{2N-k}\right)\omega_0\omega_{2N+1}}{c^3\left(q^{2N}a_0b\omega_{2N+1}+\left(\displaystyle\prod_{k=1}^{2N-1}{a_k}^{2N-k}\right)c\omega_1\right)}&\text{if } j=2N+1,
 \end{cases}\\
 &w_1(\omega_j)=\dfrac{1}{\omega_{2N+1-j}},\quad 
 j=0,\dots,2N+1.
\end{align}
\end{subequations}
Then, the statement follows from the relation between the $f$ variables and the $\omega$ variables given in \eqref{eqns:def_para_f_odd}.
\end{proof}

\begin{remark}
In the case $N=1$, the action of $\widetilde{W}((A_2\rtimes A_1)^{(1)})=\langle s_0,s_1,s_2,w_0,w_1\rangle\rtimes\langle r\rangle$ on the parameters $\{a_0,a_1,a_2,b,c,q\}$ and the variables $\{f_1,f_2\}$ corresponds to the action of $\widetilde{\bm W}((A_2\rtimes A_1)^{(1)})=\langle {\bm w}_0,{\bm w}_1,{\bm w}_2,{\bm r}_0,{\bm r}_1\rangle\rtimes\langle {\bm \pi}\rangle$ on the parameters $\{{\bm b}_0,{\bm b}_1,{\bm b}_2,{\bm b}_3,{\bm p}\}$ and the variables $\{{\bm f}_1^{(1)},{\bm f}_1^{(2)},{\bm f}_1^{(3)}\}$ in \cite{JNS2016:MR3584386} by the following correspondence:
\begin{subequations}
\begin{align}
 &{\bm w}_0=s_0,\quad
 {\bm w}_1=s_2,\quad
 {\bm w}_2=s_1,\quad
 {\bm r}_0=w_1,\quad
 {\bm r}_1=w_0,\quad
 {\bm \pi}=r,\\
 &{\bm b}_0={a_1}^3{a_2}^3,\quad
 {\bm b}_1={a_2}^3,\quad
 {\bm b}_2=a_1{a_2}^2b,\quad
 {\bm b}_3=c^2,\quad
 {\bm p}=q^{-3},\\
 &{\bm f}_1^{(1)}=\dfrac{{a_1}^3c^3(q{a_0}^2a_2b+c f_2)}{{a_2}^3f_1},\quad
 {\bm f}_1^{(2)}=\dfrac{{a_1}^3(a_1{a_2}^2bc+f_1)}{{a_2}^3c^4f_2},\quad
 {\bm f}_1^{(3)}=-\dfrac{{a_0}^3}{{a_1}^3}f_2.
\end{align}
\end{subequations}
As shown in \cite{JNS2016:MR3584386}, $\widetilde{\bm W}((A_2\rtimes A_1)^{(1)})$ is a subgroup of the extended affine Weyl symmetry group of type $A_4^{(1)}$ for Sakai's $A_4^{(1)}$-surface, which gives the $q$-Painlev\'e V equation \eqref{eqn:intro_dPV}.
Note that $\widetilde{\bm W}((A_2\rtimes A_1)^{(1)})$ is an extension of the semi-direct product of the transformation groups $\langle {\bm w}_0,{\bm w}_1,{\bm w}_2\rangle$ and $\langle {\bm r}_0,{\bm r}_1\rangle$ by the transformation ${\bm \pi}$. 
Furthermore, $\langle {\bm w}_0,{\bm w}_1,{\bm w}_2\rangle$ and $\langle {\bm r}_0,{\bm r}_1\rangle$ respectively form the affine Weyl group of type $A_2^{(1)}$ and that of type $A_1^{(1)}$, and the transformation ${\bm \pi}$ corresponds to reflections of the Dynkin diagram of type $A_2^{(1)}$ associated with $\langle {\bm w}_0,{\bm w}_1,{\bm w}_2\rangle$ and that of type $A_1^{(1)}$ associated with $\langle {\bm r}_0,{\bm r}_1\rangle$.
\end{remark}

Let us define the transformations $T_i$, $i=1,\dots,2N+1$, by \eqref{eqn:Ti_Weyl} and the transformation $\hT_0$ by
\begin{equation}
 \hT_0=\pi^{-2N-1}.
\end{equation}
The action of $T_1,\dots,T_{2N+1},\hT_0$ on the parameters $\{a_0,\dots,a_{2N},b,c,q\}$ is given by
\begin{subequations}
\begin{align}
 &T_i(a_j)
 =\begin{cases}
 qa_{i-1}&\text{if } j=i-1\quad ({\rm mod}~2N+1),\\
 q^{-1}a_i&\text{if } j=i\quad ({\rm mod}~2N+1),\\
 a_j&\text{otherwise},
 \end{cases}\qquad
 i=1,\dots,2N+1,\\
 &T_i:(b,c,q)\to(q b,c^{-1},q),\quad
 i=1,\dots,2N+1,\\
 &\hT_0:(a_0,\dots,a_{2N},b,c,q)\to(a_0,\dots,a_{2N},q^{-2N-1}b,c^{-1},q).
\end{align}
\end{subequations}

\begin{lemma}
The following holds:
\begin{equation}\label{eqn:dP_odd}
 \dfrac{{a_i}^{2N+1}\Big(\hT_0(f_i)f_i-1\Big)}{{a_i}^{2N+1}-c^{2(-1)^i}f_i}
 =\begin{cases}
 ~ \dfrac{\hT_0(f_{i+1})f_{i+1}-1}{1-{a_{i+1}}^{2N+1}c^{2(-1)^i}\hT_0(f_{i+1})}
 &\text{if } i=1,\dots,2N-1,\\[1em]
 ~ \dfrac{\left(\displaystyle\prod_{k=1}^{2N}{a_k}^k\right)bc}{\displaystyle\prod_{k=1}^N f_{2k-1}}&\text{if } i=2N.
 \end{cases}
\end{equation}
\end{lemma}
\begin{proof}
The following holds:
\begin{align}
 \pi^2(f_{2N})f_1-1
 &=\pi\left(
 -\dfrac{c^4\left(\displaystyle\prod_{k=1}^N f_{2k}\right)}{\displaystyle\prod_{k=1}^N f_{2k-1}}
 -\dfrac{q^{2N}a_0bc^3}{\left(\displaystyle\prod_{k=1}^{2N-1}{a_k}^{2N-k}\right)\left(\displaystyle\prod_{k=1}^N f_{2k-1}\right)}
 \right)f_1-1\notag\\
 &=-\left(
 \left(\displaystyle\prod_{k=1}^{N-1} f_{2k+1}\right)\pi(f_{2N})
 +\dfrac{q^{2N+1}a_1bc}{\displaystyle\prod_{k=1}^{2N-1}{a_{k+1}}^{2N-k}}
 \right)
 \dfrac{f_1}{c^4\left(\displaystyle\prod_{k=1}^N f_{2k}\right)}-1\notag\\
 &=\left(
 \dfrac{a_0c^2}{\displaystyle\prod_{k=1}^{2N-1}{a_k}^{2N-k}}
 -\dfrac{qa_1 f_1}{\displaystyle\prod_{k=1}^{2N-1}{a_{k+1}}^{2N-k}}
 \right)
 \dfrac{q^{2N}b}{c^3\left(\displaystyle\prod_{k=1}^N f_{2k}\right)}\notag\\ 
 &=\left(
 \dfrac{c^2}{\displaystyle\prod_{k=1}^{2N-1}{a_k}^{2N-k}}
 -\dfrac{{a_1}^2 f_1}{\displaystyle\prod_{k=1}^{2N-2}{a_{k+1}}^{2N-k-1}}
 \right)
 \dfrac{q^{2N}a_0b}{c^3\left(\displaystyle\prod_{k=1}^N f_{2k}\right)}\notag\\
 &=\dfrac{q^{2N}a_0b(1-{a_1}^{2N+1}c^{-2}f_1)}{c\left(\displaystyle\prod_{k=1}^{2N-1}{a_k}^{2N-k}\right)\left(\displaystyle\prod_{k=1}^N f_{2k}\right)},
\end{align}
which gives
\begin{equation}\label{eqn:dP_2N+1_proof_1}
 \dfrac{1-{a_1}^{2N+1}c^{-2}f_1}{\pi^2(f_{2N})f_1-1}
 =\dfrac{c}{q^{2N}a_0b}\left(\displaystyle\prod_{k=1}^{2N-1}{a_k}^{2N-k}\right)\left(\displaystyle\prod_{k=1}^N f_{2k}\right).
\end{equation}
Applying $\pi$ to \eqref{eqn:dP_2N+1_proof_1}, we obtain
\begin{align}\label{eqn:dP_2N+1_proof_2}
 \dfrac{1-{a_2}^{2N+1}c^2f_2}{\pi^3(f_{2N})f_2-1}
 &=\dfrac{1}{q^{2N+1}a_1bc}\left(\displaystyle\prod_{k=1}^{2N-1}{a_{k+1}}^{2N-k}\right)\left(\displaystyle\prod_{k=1}^{N-1} f_{2k+1}\right)\pi(f_{2N})\notag\\
 &=\dfrac{1}{q^{2N+1}a_1bc}
 \left(\displaystyle\prod_{k=2}^{2N}{a_{k}}^{2N-k+1}\right)
 \left(\displaystyle\prod_{k=1}^{N-1} f_{2k+1}\right)
 \pi(f_{2N})\notag\\
 &=-\dfrac{c^3}{q^{2N+1}a_1bf_1}\left(\displaystyle\prod_{k=2}^{2N}{a_{k}}^{2N-k+1}\right)\left(\displaystyle\prod_{k=1}^N f_{2k}\right)
 -\dfrac{c^2}{{a_1}^{2N+1}f_1}\notag\\
 &=\dfrac{1-{a_1}^{-2N-1}c^2\pi^2(f_{2N})}{\pi^2(f_{2N})f_1-1}.
\end{align}
Moreover, applying the transformation $w_1$ to \eqref{eqn:dP_2N+1_proof_1} and \eqref{eqn:dP_2N+1_proof_2}, we obtain
\begin{subequations}
 \begin{align}
 &\dfrac{1-{a_{2N}}^{-2N-1}c^2f_{2N}}{\pi^{-2}(f_1)f_{2N}-1}
 =\dfrac{a_0\left(\displaystyle\prod_{k=1}^N f_{2N+1-2k}\right)}{qbc\left(\displaystyle\prod_{k=1}^{2N-1}{a_{2N+1-k}}^{2N-k}\right)}
 =\dfrac{\displaystyle\prod_{k=1}^N f_{2k-1}}{\left(\displaystyle\prod_{k=1}^{2N}{a_k}^k\right)bc},\\
 &\dfrac{c^2-{a_{2N-1}}^{-2N-1}f_{2N-1}}{\pi^{-3}(f_1)f_{2N-1}-1}
 =\dfrac{c^2-{a_{2N}}^{2N+1}\pi^{-2}(f_1)}{\pi^{-2}(f_1)f_{2N}-1},
\end{align}
\end{subequations}
respectively.
Then, using 
\begin{equation}
 \pi^{-2}(f_1)=\hT_0(f_{2N}),\quad
 \pi^{-3}(f_1)=\hT_0(f_{2N-1}),
\end{equation}
we obtain
\begin{subequations}
 \begin{align}
 &\dfrac{{a_{2N}}^{2N+1}\Big(\hT_0(f_{2N})f_{2N}-1\Big)}{{a_{2N}}^{2N+1}-c^2f_{2N}}
 =\dfrac{\left(\displaystyle\prod_{k=1}^{2N}{a_k}^k\right)bc}{\displaystyle\prod_{k=1}^N f_{2k-1}},\\
 &\dfrac{{a_{2N-1}}^{2N+1}\Big(\hT_0(f_{2N-1})f_{2N-1}-1\Big)}{{a_{2N-1}}^{2N+1}-c^{-2}f_{2N-1}}
 =\dfrac{\hT_0(f_{2N})f_{2N}-1}{1-{a_{2N}}^{2N+1}c^{-2}\hT_0(f_{2N})}.
 \label{eqn:dP_2N+1_proof_3}
\end{align}
\end{subequations}
The equations above are the $i=2N,\,2N-1$ cases of \eqref{eqn:dP_odd}, respectively.
Then the statement follows by applying the transformations $\pi^{-m}$, $m=1,\dots,2N-2$, to Equation \eqref{eqn:dP_2N+1_proof_3}.
\end{proof}

The system \eqref{eqn:dP_odd} is equivalent to the system \eqref{eqn:intro_dP_odd} by the following correspondence:
\begin{equation}\label{eqn:correspondence_introEqn_n=odd}
 \overline{\rule{0em}{0.5em}~\,}=\hT_0,\quad
 G_i=f_i,\quad
 t=b,\quad
 p=q^{-2N-1}.
\end{equation}
Therefore, from this fact and Lemma \ref{lemma:n=odd_f_Weyl}, Theorem \ref{theorem:symmetry_n=odd} holds.

\section{Lax pairs of the $q$P$^{(2N)}(A_6^{(1)})$ \eqref{eqn:intro_dP_even} and the $q$P$^{(2N)}(A_4^{(1)})$ \eqref{eqn:intro_dP_odd}}\label{section:Lax_eqns}
In this section, we prove Theorems \ref{theorem:lax_n=even} and \ref{theorem:lax_n=odd}, that is, we construct Lax pairs of the systems \eqref{eqn:intro_dP_even} and \eqref{eqn:intro_dP_odd}. 
For detailed construction methods of a Lax pair of a Painlev\'e type difference equation from a higher dimensional CAC system using a periodic-reduction, see, for example, \cite{JN2016:MR3597921,JNS2016:MR3584386}. 

\subsection{Proof of Theorem \ref{theorem:lax_n=even} with $N=1$}\label{subsection:proof_Lax_n=2}
Under the conditions \eqref{eqn:period_U_2} and \eqref{eqn:period_para_2}, we consider the Lax equations \eqref{eqns:Weyl_phi} when $n=2$.
Define the column vector of length two $\Phi$, the spectral variable $x$, and the spectral operator $T_x$ as
\begin{equation}
 \Phi
 =\begin{pmatrix}
 {\omega_0}^{-2}&0\\0&1
 \end{pmatrix}
 \phi(0,0,0),\quad
 x=\al^{(1)}(0)^{-1/2}\al^{(2)}(0)^{-1/2}\mu,\quad
 T_x=T_1T_2T_0,
\end{equation}
where $T_i$, $i=0,1,2$, are given in Remarks \ref{remark:Ti_U} and \ref{remark:Ti_phi} and
\begin{equation}
 \omega_0=U(0,0,0).
\end{equation}
From \eqref{eqns:Weyl_phi}, we have
\begin{subequations}
\begin{align}
 &T_i(\Phi)
 =\begin{pmatrix}
 \dfrac{\mu \omega_0}{\al^{(i)}(0)T_i(\omega_0)}&-1\\1&-\dfrac{\mu T_i(\omega_0)}{\al^{(i)}(0)\omega_0}
 \end{pmatrix}
 \Phi,\quad
 i=1,2,\\
 &T_0(\Phi)
 =\begin{pmatrix}
 -\dfrac{\mu\ka(0)\omega_0}{T_0(\omega_0)}&-1\\1&0
 \end{pmatrix}
 \Phi.
\end{align}
\end{subequations}
Therefore, we obtain
\begin{subequations}\label{eqns:TxT0_Phi_n=2}
\begin{align}
 &T_x(\Phi)
 =\begin{pmatrix}
 -\dfrac{b x}{f_1f_2}&-1\\1&0
 \end{pmatrix}
 \begin{pmatrix}
 \dfrac{x f_2}{a_1}&-1\\1&-\dfrac{x}{a_1 f_2}
 \end{pmatrix}
 \begin{pmatrix}
 a_1x f_1&-1\\1&-\dfrac{a_1x}{f_1}
 \end{pmatrix}
 \Phi,\\
 &T_0(\Phi)
 =\begin{pmatrix}
 -\dfrac{b x}{T_0(f_1)T_0(f_2)}&-1\\1&0
 \end{pmatrix}
 \Phi.
\end{align}
\end{subequations}
The action of $T_x$ on $\{a_0,a_1,b,q,f_1,f_2\}$ is given by
\begin{equation}
 T_x:(a_0,a_1,b,q,f_1,f_2)\to(a_0,a_1,b,q,f_1,f_2),
\end{equation}
which is trivial, but that on the spectral parameter $x$ is not trivial as the following:
\begin{equation}
 T_x(x)=q^{-2}x.
\end{equation}
On the other hand, the action of $T_0$ on $\{a_0,a_1,b,q,x\}$ is given by
\begin{equation}
 T_0:(a_0,a_1,b,q,x)\to(a_0,a_1,q^{-2}b,q,x),
\end{equation}
and that on $\{f_1,f_2\}$ is given by \eqref{eqn:T0_f_n=2}.
Therefore, Theorem \ref{theorem:lax_n=even} with $N=1$ follows from \eqref{eqn:correspondence_introEqn_n=2}, \eqref{eqns:TxT0_Phi_n=2} and the following correspondence:
\begin{equation}\label{eqn:Phi_correspondence_n=even}
 \Phi(x,t)=\Phi,\quad
 \Phi(px,t)=T_x(\Phi),\quad
 \Phi(x,pt)=T_0(\Phi).
\end{equation}

\begin{remark}
As mentioned in \S \ref{subsection:reduction_n=2}, the action of $\hT_0$ and that of $T_0$ on the parameters $\{a_0,a_1,b,q\}$ and the variables $\{f_1,f_2\}$ are same, but their actions on the parameters $\{\al^{(1)}(0),\al^{(2)}(0),\ka(0)\}$ are different. 
This difference gives their different actions on the spectral parameter $x$ as follows.
\begin{equation}
 \hT_0(x)=q^2x,\quad
 T_0(x)=x,
\end{equation}
which leads to the spectral operator $T_x$.
Indeed, $T_x$ can also be written as
\begin{equation}
 T_x={\hT_0}^{~-1}T_0.
\end{equation}
This idea was proposed in \cite{JN2016:MR3597921}.
It is an extension of the ideas presented in \cite{NP1991:MR1098879,GRSWC2005:MR2117991,FJN2008:MR2425981,HHJN2007:MR2303490}, where Lax pairs of discrete Painlev\'e equations are constructed from two-dimensional partial difference equations by using staircase methods.
\end{remark}

\subsection{Proof of Theorem \ref{theorem:lax_n=even} with $N\in\bbZ_{\geq2}$}\label{subsection:proof_Lax_n=even}
Under the conditions \eqref{eqn:period_U_even} and \eqref{eqn:period_para_even}, we consider the Lax equation \eqref{eqns:Weyl_phi} when $n=2N$ $(N=2,3,\dots)$.
The discussion is the same as in \S \ref{subsection:proof_Lax_n=2}, so the details are omitted.

We define the column vector of length two $\Phi$, the spectral variable $x$, and the spectral operator $T_x$ as
\begin{equation}
 \Phi
 =\begin{pmatrix}
 {\omega_0}^{-2}&0\\0&1
 \end{pmatrix}
 \phi(0,\dots,0),\quad
 x=\left(\prod_{k=1}^{2N}\al^{(k)}(0)^{-1/(2N)}\right)\mu,\quad
 T_x=T_1\cdots T_{2N}T_0,
\end{equation}
where $T_i$, $i=0,\dots,2N$, are given in Remarks \ref{remark:Ti_U} and \ref{remark:Ti_phi} and
\begin{equation}
 \omega_0=U(0,\dots,0).
\end{equation}
From \eqref{eqns:Weyl_phi}, we have
\begin{subequations}
\begin{align}
 &T_i(\Phi)
 =\begin{pmatrix}
 \dfrac{\mu \omega_0}{\al^{(i)}(0)T_i(\omega_0)}&-1\\1&-\dfrac{\mu T_i(\omega_0)}{\al^{(i)}(0)\omega_0}
 \end{pmatrix}
 \Phi,\quad
 i=1,\dots,2N,\\
 &T_0(\Phi)
 =\begin{pmatrix}
 -\dfrac{\mu\ka(0)\omega_0}{T_0(\omega_0)}&-1\\1&0
 \end{pmatrix}
 \Phi.
\end{align}
\end{subequations}
Therefore, we obtain
\begin{equation}\label{eqn:TxT0_Phi_n=even}
 T_x(\Phi)
 =\begin{pmatrix}
 -\dfrac{b x}{~\displaystyle\prod_{k=1}^{2N}f_k~}&-1\\1&0
 \end{pmatrix}
 L_{2N}\dots L_1
 \Phi,\quad
 T_0(\Phi)
 =\begin{pmatrix}
 -\dfrac{b x}{~\displaystyle\prod_{k=1}^{2N}T_0(f_k)~}&-1\\1&0
 \end{pmatrix}
 \Phi,
\end{equation}
where
\begin{equation}
 L_i=\begin{pmatrix}
 \dfrac{\left(\displaystyle\prod_{k=1}^{2N-i}{a_{2N-k}}^k\right)xf_i}{\displaystyle\prod_{k=1}^{i-1}{a_k}^k}&-1\\
 1&-\dfrac{\left(\displaystyle\prod_{k=1}^{2N-i}{a_{2N-k}}^k\right)x}{\left(\displaystyle\prod_{k=1}^{i-1}{a_k}^k\right)f_i}
 \end{pmatrix},\quad
 i=1,\dots,2N.
\end{equation}
The actions of $T_x$ and $T_0$ on the parameters $\{a_0,\dots,a_{2N-1},b,x,q\}$ are given by
\begin{subequations}
\begin{align}
 &T_x:(a_0,\dots,a_{2N-1},b,x,q)\to (a_0,\dots,a_{2N-1},b,q^{-2N}x,q),\\
 &T_0:(a_0,\dots,a_{2N-1},b,x,q)\to (a_0,\dots,a_{2N-1},q^{-2N}b,x,q).
\end{align}
\end{subequations}
Note that the action of $T_x$ on the $f$ variables is trivial, and that of $T_0$ is given by \eqref{eqn:dP_even}.
Therefore, Theorem \ref{theorem:lax_n=even} with $N\in\bbZ_{\geq2}$ follows from \eqref{eqn:TxT0_Phi_n=even} and the correspondences \eqref{eqn:correspondence_introEqn_n=even} and \eqref{eqn:Phi_correspondence_n=even}.

\subsection{Proof of Theorem \ref{theorem:lax_n=odd}}\label{subsection:proof_Lax_n=odd}
Under the conditions \eqref{eqn:period_U_odd} and \eqref{eqn:period_para_odd}, we consider the Lax equation \eqref{eqns:Weyl_phi} when $n=2N+1$ $(N=1,2,\dots)$.
The discussion is the same as in \S \ref{subsection:proof_Lax_n=2}, so the details are omitted.
Note that the result for $N=1$ is the same as that in \cite{JNS2016:MR3584386}.

We define the column vector of length two $\Phi$, the spectral variable $x$, and the spectral operator $T_x$ as
\begin{subequations}
\begin{align}
 &\Psi
 =\begin{pmatrix}
 {\omega_{2N+1}}^{-1}&0\\0&\omega_0
 \end{pmatrix}
 \phi(0,\dots,0,-1),\quad
 x=\left(\prod_{k=1}^{2N+1}\al^{(k)}(0)^{-1/(2N+1)}\right)\mu,\\
 &T_x=T_0T_1\cdots T_{2N+1},
\end{align}
\end{subequations}
where $T_i$, $i=0,\dots,2N+1$, are given in Remarks \ref{remark:Ti_U} and \ref{remark:Ti_phi}.
Here, the variables $\omega_i$, $i=0,\dots,2N+1$, are given by \eqref{eqn:def_omega_odd}.
From \eqref{eqns:Weyl_phi}, we have
\begin{subequations}
\begin{align}
 &
 T_i(\Psi)
 =
 \begin{pmatrix}
 \dfrac{\mu }{\al^{(i)}(0)}&-\dfrac{cT_i(\omega_{2N+1})}{\omega_0}\\[1em]
 \dfrac{T_i(\omega_0)}{c\omega_{2N+1}}&-\dfrac{\mu T_i(\omega_0)T_i(\omega_{2N+1})}{\al^{(i)}(0)\omega_0\omega_{2N+1}}
 \end{pmatrix}
 \Psi,\quad
 i=1,\dots,2N+1,\\
 &T_0(\Psi)
 =\begin{pmatrix}
 -\mu\ka(-1)&-\dfrac{1}{c^2}\\[1em]
 \dfrac{c^2T_0(\omega_0)}{\omega_{2N+1}}&0
 \end{pmatrix}
 \Psi.
\end{align}
\end{subequations}
Therefore, we obtain
\begin{equation}
 T_x(\Psi)=L_{2N+1}\dots L_0\Psi,\quad
 T_0(\Psi)
 =L_0
 \Psi,
\end{equation}
where
\begin{subequations}
\begin{align}
 L_0
 &=\begin{pmatrix}
 -\mu\ka(-1)&-\dfrac{1}{c^2}\\[1em]
 \dfrac{c^2T_0(\omega_0)}{\omega_{2N+1}}&0
 \end{pmatrix}
 =-\begin{pmatrix}\dfrac{1}{c}&0\\0&-\dfrac{T_0(\omega_0)}{\omega_1}\end{pmatrix}
 \begin{pmatrix}\mu\ka(-1)c&\dfrac{1}{c}\\\dfrac{c^2\omega_1}{\omega_{2N+1}}&0\end{pmatrix},\\
 L_i
 &=\begin{pmatrix}
 \dfrac{\mu}{\al^{(i)}(0)}&-\dfrac{c^{(-1)^i}\omega_i}{T_0(\omega_{i-1})}\\[1em]
 \dfrac{T_0(\omega_i)}{c^{(-1)^i}\omega_{i-1}}&-\dfrac{\mu\omega_iT_0(\omega_i)}{\al^{(i)}(0)\omega_{i-1}T_0(\omega_{i-1})}
 \end{pmatrix}\notag\\
 &=\begin{pmatrix}c^{(-1)^i}&0\\0&\dfrac{T_0(\omega_i)}{\omega_{i-1}}\end{pmatrix}
 \begin{pmatrix}\dfrac{\mu}{\al^{(i)}(0)}&1\\1&\dfrac{\mu}{\al^{(i)}(0)}\end{pmatrix}
 \begin{pmatrix}\dfrac{1}{c^{(-1)^i}}&0\\0&-\dfrac{\omega_i}{T_0(\omega_{i-1})}\end{pmatrix},\notag\\
 &\quad i=1,\dots,2N+1.
\end{align}
\end{subequations}
In the following, let us rewrite the equations above in notation with the parameters $\{a_0,\dots,a_{2N},b,c,x\}$ and the $f$ variables.
Letting
\begin{subequations}
\begin{align}
 &M_i=\begin{pmatrix}\dfrac{\mu}{\al^{(i)}(0)}&1\\1&\dfrac{\mu}{\al^{(i)}(0)}\end{pmatrix}
 =\begin{pmatrix}
 \dfrac{\displaystyle\prod_{k=i}^{2N}{a_k}^{2N+1}}{\displaystyle\prod_{k=1}^{2N}{a_k}^k}x&1\\
 1&\dfrac{\displaystyle\prod_{k=i}^{2N}{a_k}^{2N+1}}{\displaystyle\prod_{k=1}^{2N}{a_k}^k}x
 \end{pmatrix}, 
 &&i=1,\dots,2N+1,\\
 &K_j=\begin{pmatrix}c^{2(-1)^j}&0\\0&-\dfrac{\omega_{j+1}}{\omega_{j-1}}\end{pmatrix}
 =\begin{pmatrix}c^{2(-1)^j}&0\\0&-\dfrac{1}{f_j}\end{pmatrix},
 &&j=1,\dots,2N,
\end{align}
\end{subequations}
we obtain 
\begin{align}\label{eqn:TxT0_Phi_n=odd_1}
 T_x(\Psi)
 =&-\begin{pmatrix}\dfrac{1}{c}&0\\0&\dfrac{\omega_0}{\omega_{2N}}\end{pmatrix}
 M_{2N+1}K_{2N}M_{2N}\dots K_{1}M_{1}
 \begin{pmatrix}\mu\ka(-1)c&\dfrac{1}{c}\\\dfrac{c^2\omega_1}{\omega_{2N+1}}&0\end{pmatrix}
 \Psi\notag\\
 =&-\begin{pmatrix}\dfrac{1}{c}&0\\0&\displaystyle\prod_{k=1}^Nf_{2k-1}\end{pmatrix}
 M_{2N+1}K_{2N}M_{2N}\dots K_{1}M_{1}
 \begin{pmatrix}q^{2N+1}bcx&\dfrac{1}{c}\\c^2\left(\displaystyle\prod_{k=1}^Nf_{2k}\right)&0\end{pmatrix}
 \Psi.
\end{align}
Moreover, from \eqref{eqn:ABS_U_0k} we obtain
\begin{equation}
-\dfrac{T_0(\omega_0)}{\omega_1}
=\dfrac{T_0(\omega_1)}{c^4\omega_0}+\dfrac{\al^{(1)}(0) \ka(0)}{c^3}
=\dfrac{1}{c^4}\left(\displaystyle\prod_{k=1}^NT_0(f_{2k})\right)+\dfrac{b}{c^3}\left(\displaystyle\prod_{k=1}^{2N}{a_k}^{k-2N-1}\right),
\end{equation}
which gives
\begin{equation}\label{eqn:TxT0_Phi_n=odd_2}
 T_0(\Psi)
 =-\begin{pmatrix}\dfrac{1}{c}&0\\
 0&\dfrac{1}{c^4}\left(\displaystyle\prod_{k=1}^NT_0(f_{2k})\right)+\dfrac{b}{c^3}\left(\displaystyle\prod_{k=1}^{2N}{a_k}^{k-2N-1}\right)
 \end{pmatrix}
 \begin{pmatrix}q^{2N+1}bcx&\dfrac{1}{c}\\c^2\left(\displaystyle\prod_{k=1}^Nf_{2k}\right)&0\end{pmatrix}
 \Psi.
\end{equation}
The actions of $T_x$ and $T_0$ on the parameters $\{a_0,\dots,a_{2N},b,c,x,q\}$ are given by
\begin{subequations}
\begin{align}
 &T_x:(a_0,\dots,a_{2N},b,c,x,q)\to (a_0,\dots,a_{2N},b,c,q^{-2N-1}x,q),\\
 &T_0:(a_0,\dots,a_{2N},b,c,x,q)\to (a_0,\dots,a_{2N},q^{-2N-1}b,c^{-1},x,q).
\end{align}
\end{subequations}
Note that the action of $T_x$ on the $f$ variables is trivial, and that of $T_0$ is given by \eqref{eqn:dP_odd}.
Therefore, Theorem \ref{theorem:lax_n=odd} follows from \eqref{eqn:correspondence_introEqn_n=odd}, \eqref{eqn:TxT0_Phi_n=odd_1}, \eqref{eqn:TxT0_Phi_n=odd_2} and the following correspondence:
\begin{equation}
 \Psi(x,t)=\Psi,\quad
 \Psi(px,t)=T_x(\Psi),\quad
 \Psi(x,pt)=T_0(\Psi).
\end{equation}
Note that the coefficient matrix in Theorem \ref{theorem:lax_n=odd} is multiplied by $(-1)$ for simplicity.
The following gauge transformation can explain this multiplication:
\begin{equation}
 \Psi\to \dfrac{\Theta(-xt;p)}{\Theta(xt;p)}\Psi,
\end{equation}
where $\Theta(a;p)$ is the modified Jacobi theta function \cite{book_GR2004:MR2128719} satisfying
\begin{equation}
 \cfrac{\Theta(pa;p)}{\Theta(a;p)}=-a^{-1}.
\end{equation}
\section{Concluding remarks}\label{ConcludingRemarks}
In this paper, we have constructed the higher-order $q$-Painlev\'e systems \eqref{eqn:intro_dP_even} and \eqref{eqn:intro_dP_odd}, which include second-order $q$-Painlev\'e equations of $A_6^{(1)}$- and $A_4^{(1)}$-surface type, respectively.
We also obtained their extended affine Weyl group symmetries and Lax pairs.

KNY's representation is well known for giving a high-dimensional extension of $q$-Painlev\'e equations \cite{KNY2002:MR1917133,KNY2002:MR1958118}.
As explained in \S \ref{subsection:Background}, it was extended from $(A_{m-1}\times A_{n-1})^{(1)}$-type to $(A_{m-1}\times A_{n-1}\times A_{g-1})^{(1 )}$-type by using the theory of cluster algebra\cite{MOT2018:AmAnAg,MOT2023:AmAnAgArxiv}.
Similarly, it is expected that the symmetries of the systems \eqref{eqn:intro_dP_even} and \eqref{eqn:intro_dP_odd} can be extended.
From the perspective of this study, we consider that the multi-component extension of CAC systems \cite{ZKZ2020:zbMATH07227205} is effective.
Other approaches based on cluster algebra \cite{MOT2018:AmAnAg,MOT2023:AmAnAgArxiv,OS2020:10.1093/imrn/rnaa283,HI2014:zbMATH06381200} and birational representation of affine Weyl group \cite{KNY2002:MR1917133,KNY2002:MR1958118,PY2021:factorized,MasudaT2015:zbMATH06618811,TT2009:MR2511044} also seem to be highly effective.
Research on this extension is a subject for future study.
\subsection*{Acknowledgment}
I would like to thank Prof. Masatoshi Noumi, Dr. Takao Suzuki, and Dr. Naoto Okubo for fruitful discussions.
This work was supported by JSPS KAKENHI Grant Numbers JP19K14559 and JP23K03145.

\def\cprime{$'$} \def\cprime{$'$}

\end{document}